\documentclass{article}
  
\usepackage{amssymb,amsthm,amsmath}
\usepackage{bbm}
\usepackage{proof}
\usepackage[all,cmtip]{xy}
\usepackage{tikz}
\usetikzlibrary{cd,calc,automata,positioning,arrows,shapes}

\usepackage{enumitem}
\usepackage{fullpage}
\usepackage[numbers]{natbib}
\usepackage{multicol}
\theoremstyle{definition}
\newtheorem{theorem}{Theorem}
\numberwithin{theorem}{section}
\newtheorem{definition}[theorem]{Definition}

\newtheorem{lemma}[theorem]{Lemma}
\newtheorem{claim}[theorem]{Claim}
\newtheorem{corollary}[theorem]{Corollary}
\newtheorem{proposition}[theorem]{Proposition}
\newtheorem{example}[theorem]{Example}
\newtheorem{remark}[theorem]{Remark}
\numberwithin{equation}{section}
\newcommand{\lequardar}{\leq^{\uar\dar}}
\newcommand{\semantics}[1]{[\![\mbox{\em $ #1 $\/}]\!]}
\newcommand{\semanticsalt}[1]{\langle\!\langle\mbox{\em $ #1 $\/}\rangle\!\rangle}

\newcommand{\FF}{\mathbb{F}}
\newcommand{\GG}{\mathbb{G}}
\newcommand{\PP}{\mathbb{P}}
\newcommand{\RR}{\mathbb{R}}
\renewcommand{\AA}{\mathbb{A}}
\newcommand{\Xsyn}{\mathrm{X}^{syn}}
\newcommand{\Xsem}{\mathrm{X}^{sem}}
\newcommand{\PPsyn}{\mathbb{P}^{syn}}
\newcommand{\PPsem}{\mathbb{P}^{sem}}

\newcommand{\QQ}{\mathbb{Q}}
\newcommand{\LL}{\mathbb{L}}
\newcommand{\MM}{\mathbb{M}}
\renewcommand{\SS}{\mathbb{S}}
\usepackage{verbatim}
\newcommand{\TT}{\mathcal{T}}

\newcommand{\Model}{\mathcal{M}}
\newcommand{\Nodel}{\mathcal{N}}

\newcommand{\set}[1]{\{ #1 \}}
\newcommand{\arrowplus}{\overset{+}{\rightarrow} }
\newcommand{\arrowminus}{\overset{-}{\rightarrow} }
\newcommand{\arrowdot}{\overset{\cdot}{\rightarrow} }
\newcommand{\arrowm}{\overset{m}{\rightarrow} }

\newcommand{\arrowo}{\overset{\circ}{\rightarrow} }
\newcommand{\arrowcirc}{\arrowo}

\newcommand{\proves}{\vdash}
\newcommand{\proveswc}{\vdash_{\mbox{\sc wc}}}

\newcommand{\modelswc}{\models_{\mbox{\sc wc}}}

\newcommand{\uar}{+}
\newcommand{\dar}{-}

\newcommand{\arrowmprime}{\overset{m'}{\rightarrow} }
\newcommand{\quadiff}{\quad \mbox{ iff } \quad}

\renewcommand{\o}{\circ}
\renewcommand{\O}{\o}

\newcommand{\Types}{\mathcal{T}}

\newcommand{\rem}[1]{\relax}
\newcommand{\GammaBox}{\Gamma_{\Box}}
\newcommand{\Mar}{\mbox{\sf Mar}}

\renewcommand{\o}{\circ}
\newcommand{\wcintrees}{\mbox{\sc wc}}
\newcommand{\wcone}{\mbox{\sc wc}${}_1$}
\newcommand{\wctwo}{\mbox{\sc wc}${}_2$}
\newcommand{\wcthree}{{\sc wc}${}_3$}
\newcommand{\congrule}{\mbox{\sc cong}}

\newcommand{\monorule}{\mbox{\sc mono}}
\newcommand{\antirule}{\mbox{\sc anti}}
\newcommand{\pointrule}{\mbox{\sc point}}
\newcommand{\polrule}{\mbox{\sc pol}}
\newcommand{\transrule}{\mbox{\sc trans}}
\newcommand{\symmrule}{\mbox{\sc symm}}
\newcommand{\weakrule}{\mbox{\sc weak}}
\newcommand{\posrule}{\mbox{\sc pos}}
\newcommand{\posruleprime}{\mbox{\sc pos}'}
\newcommand{\wcrule}{\mbox{\sc wc}}
\newcommand{\reflrule}{\mbox{\sc refl}}

\newlength{\mathfrwidth}
\newsavebox{\mathfrbox}
\newenvironment{mathframe}
    {\begin{lrbox}{\mathfrbox}\begin{minipage}{\mathfrwidth}\begin{center}}
    {\end{center}\end{minipage}\end{lrbox}\noindent\fbox{\usebox{\mathfrbox}}}

 \renewcommand{\phi}{\varphi}
 \rem{
 \newbox\qqBoxA
\newdimen\qqCornerHgt
\setbox\qqBoxA=\hbox{$\ulcorner$}
\global\qqCornerHgt=\ht\qqBoxA
\newdimen\qqArgHgt
\def\Qq #1{%
    \setbox\qqBoxA=\hbox{$#1$}%
    \qqArgHgt=\ht\qqBoxA%
    \ifnum     \qqArgHgt<\qqCornerHgt \qqArgHgt=0pt%
    \else \advance \qqArgHgt by -\qqCornerHgt%
    \fi \raise\qqArgHgt\hbox{$\ulcorner$} \box\qqBoxA %
    \raise\qqArgHgt\hbox{$\urcorner$}}
 }
 \newcommand{\Qq}[1]{q_{#1}} 
 \newcommand{\OP}{\mbox{\emph{OP}}}
 
\renewcommand{\arrowdot}{\to} 
\newcommand{\opp}{\mbox{\scriptsize \sf op}}  
\begin{document}

\title{A Completeness Result for Inequational Reasoning\\
 in a Full Higher-Order Setting}
\author{
  Lawrence S. Moss\thanks{Department of Mathematics, Indiana University, Bloomington IN 47401,  {\tt lmoss@indiana.edu}.
  This work was partially supported by a grant from the Simons Foundation ($\#$245591 to Lawrence Moss).}
 \and
 Thomas F. Icard\thanks{Department of Philosophy, Stanford University, Stanford, CA, USA,   {\tt icard@stanford.edu}} \and
}    
\date{ }

\maketitle

\begin{mathframe}
This paper will appear in a volume in College Publications' Tributes series, edited by Katalin Bimb\'{o}.\\
It is dedicated to the memory of J.~Michael Dunn.
\end{mathframe}

\begin{abstract}
This paper obtains a completeness result for inequational
 reasoning with applicative terms without variables in a setting where the 
intended semantic models are the full structures, the full type hierarchies over preorders for the base types.
The syntax allows for the 
specification that a given constant be interpreted as a monotone function, or an antitone function, or both.
There is a natural set of five rules for this inequational reasoning. 
One can add variables and also add a substitution rule, but we observe that this logic would be incomplete for full structures.
This is why the completeness result in this paper pertains to terms without variables.   
Since the completeness is already known for the class of general (Henkin) structures, 
we are interested in full structures.
We  obtain the first  
result on this topic.
Our result is not optimal because we restrict to base preorders which have a weak completeness property: 
every pair of elements has an
upper bound and a lower bound.  To compensate we add several rules to the logic.
We also present extensions and variations of our completeness result.
\end{abstract}


\section{Introduction}

\paragraph{Tonoids recast}
In his work on very general algebraic semantics of non-classical logics, Dunn~\cite{Dunn93} introduces the notion of a 
\emph{tonoid}.  This is a structure of the form $(A,\leq, \OP)$, where $\AA = (A,\leq)$ is a poset, and $\OP$ is a set of 
finite-arity \emph{function symbols}, each with a \emph{tonic type} $(s_1,\ldots, s_n)$, where each $s_i$ is either $+$ or $-$.
A familiar example done this way takes $\AA$ to be $2 = \set{0,1}$ with $0 < 1$, and $\OP = \set{\to}$, where $\to$
is taken as an operation with tonicity type $(-,+)$.  
The formal requirement is that if $f\in \OP$ is of arity $n$, then $f: A^n\to A$ is either isotone or antitone in the $i$th argument, depending 
on whether $s_i$ is $+$ or $-$.    To spell out the requirement in more detail, recall that a function $g:A\to A$ is \emph{isotone}
(here called \emph{monotone}) if $a\leq b$ implies $g(a) \leq g(b)$; and 
 $g:A\to A$ is \emph{antione} if $a\leq b$ implies $g(b) \leq g(a)$.   Suppose that $f$ is of arity $3$ and its tonic type is $(+,-,+)$.
 Then our requirement is:
\begin{equation}\label{req1}
\mbox{whenever $a_1 \leq a_2$, $b_2 \leq b_1$, and $c_1\leq c_2$, $f(a_1,b_1,c_1) \leq f(a_2, b_2, c_2)$}.
\end{equation}
The idea is to abstract a feature of material implication:
it is antitone in its first argument
and monotone in its second.
Here are two equivalent ways to state the general requirement (\ref{req1}).
The first uses the concept of the \emph{opposite} poset $\AA^{\opp}$; this is $\AA$
with the converse order.  Our requirement (\ref{req1}) now would say that 
\begin{equation}\label{req2}
f: \AA \times \AA^{\opp} \times \AA \to \AA .
\end{equation}
In this, $\times$ denotes the product operation on posets,  and the arrow $\to$ means ``monotone function.''
This formulation (\ref{req2})
 can be recast by currying,  replacing a function of arity $3$ by a higher-order function of the following form:
\begin{equation}\label{req3}
f: \AA \to \AA^{\opp} \to \AA \to \AA  .
\end{equation}
 So $\AA\to \AA$ is the set of monotone functions,
taken as a poset $\PP$ with the pointwise order. Then $\AA^{\opp} \to \AA\to \AA$ is 
the set of monotone functions from $\AA^{\opp}$ to $\PP$, again taken as a poset which we call $\QQ$.
  Finally, $\AA \to \AA^{\opp} \to \AA \to \AA$ is the set of monotone functions from $\AA$ to $\QQ$.
Going one step further from (\ref{req3}), our requirement may be rephrased once again.
\begin{equation}\label{req4}
f: \AA \arrowplus \AA\arrowminus \AA \arrowplus \AA  .
\end{equation}
In (\ref{req4}), the operative notation is that $\PP\arrowplus\QQ$ denotes the set of monotone functions from $\PP$ to $\QQ$,
and $\PP\arrowminus\QQ$ denotes the set of antitone functions from $\PP$ to $\QQ$.  In both cases, the order is pointwise.

Up until now, all we have done is to rephrase the definition of a tonoid in terms of higher-order functions in the realm of posets,
something that Dunn did not need to do.   We are indeed interested exactly in higher-order reasoning about ordered structures.
Instead our result is aimed at settings where reasoning about monotone/antitone functions plays a central role.  
One such setting is the area of programming language semantics where the order represents subtyping.
Another is natural language inference where higher-order functions are commonplace, following the tradition in 
Montague grammar and type-logical grammar.  Concerning inference, van Benthem~\cite{vanBenthemEssays86}
pointed out the usefulness of monotonicity in connection with the higher-order semantics of determiners
and saw that this topic would be a central part of logical studies connected to natural language.  
The connection to higher-order preorders in this area was first made in~\cite{moss:polarity},
and that paper is also the source of the observations behind the moves from (\ref{req1}) to (\ref{req4}).

\paragraph{Friedman's Theorem on the STLC}
The results that we are after in this paper are modeled on the completeness result established
 by Friedman~\cite{Friedman1975}
for the simply-typed lambda calculus (STLC).
To explain our contribution, let us review part of Friedman's contribution.
We change the notation and presentation of~\cite{Friedman1975} to set the stage for our work.

The STLC begins with a set $B$ of base types $\beta$.   The full set $\Types$ of types is the closure of $B$ under
the following rule: if $\sigma$ and $\tau$ are types, so is $\sigma\to\tau$.   Then one forms the set of typed terms
$t:\sigma$
 of the STLC
using application of one term to another, variables, and abstraction.
The main assertions in the STLC are identities $t=u$ between terms of the same type.
The semantics is of interest here.   The primary models are  \emph{full} (or \emph{standard}) type structures.
Beginning with sets $X_{\beta}$ for $\beta\in B$, one constructs sets $X_{\sigma}$ for all types $\sigma\in\Types$ by recursion:
$X_{\sigma\to\tau}$ is the set 
$(X_{\tau})^{X_{\sigma}}$ 
of \emph{all} functions from $X_{\sigma}$ to $X_{\tau}$.   Then one interprets each typed term $t:\sigma$ by an element $\semantics{t}\in X_{\sigma}$.
Naturally, one is interested in the relation on terms $\models t = u $ defined by:
\begin{equation}
\label{eq-consequence}
\models t = u \mbox{ iff } \semantics{t} = \semantics{u} \mbox{ in all full structures}
\end{equation}
The main completeness result from~\cite{Friedman1975} is that $\models t = u $ iff the statement $t = u$ can be proved in a certain logical system 
with very natural rules.
The rules of the system are the reflexive, symmetric, and transitive rules of identity, the congruence 
 rule for application, and the $\alpha$, $\beta$, and $\eta$ rules of the STLC.   So the completeness of the system tells us that 
 an identity assertion holds in all full structures iff it is provable
 from  $\alpha$, $\beta$, and $\eta$ on top of the expected rules of identity.

\paragraph{What we are doing}
Here is how things change in this paper.
We would like the main assertions in our system to be \emph{inequalities} $t\leq u$ instead of identites.
Thus, we want our semantic spaces to be \emph{preorders} rather than unstructured sets.
  Beginning with an assignment of preorders $(\PP_\beta)_{\beta\in B}$ to the base types,
we take preorders for function types $\PP_{\sigma\to\tau}$ to be the set of all functions, as above, but endowed with the 
pointwise order.
This is what we mean by \emph{full} models in our title and throughout the paper.
   Moreover, we allow our type system to insist that a given function symbol be interpreted by a 
monotone function, or an antitone one (or both).  

The logical systems in this paper are formulated \emph{without variables}: the only terms are those which can be 
constructed from the typed \emph{constants} using application.    This might seem to be a severe limitation,
so let us motivate it from several angles.   First of all, monotonicity calculi without variables 
are useful in several settings (see~\cite{IcardMoss2014}).
Second, the completeness results of interest in this paper are not available if one has variables
(see Section~\ref{section-digression}).
This is a parallel to the matter of equational reasoning with second-order terms (even without abstraction):
 the natural logical system would add substitution to the rules mentioned above.
  This system is not complete for full models.
Finally, the authors and William Tune have formulated ``order-aware'' versions of the lambda calculus (see~\cite{IcardMoss21,IMT,Tune2016}).
The 
type system expands that of the usual simply typed lambda calculus by permitting the formation of 
several additional kinds of function types: monotone functions $\sigma\arrowplus\tau$, antitone functions $\sigma\arrowminus\tau$,
and others.  
Tune~\cite{Tune2016} is a variation on this which incorporates something like the ``{\sf op}'' operation on preorders which we have seen above in
(\ref{req2}).
All work in this area expands the syntax of terms using variables and abstraction operations.  Finally, the basic assertions in the language
include inequalities between objects of the same type.  What is more, it includes some inequalities between objects of
different (but related) types.   The formulation of the syntax is non-trivial, and the same goes for the proof rules.
  In any case, as we already mentioned,
the logical systems in the area cannot be complete for full models.   They are complete for wider classes of 
``Henkin'' models. 
(The analogous structures for the STLC
 in~\cite{Friedman1975} are called \emph{pre-structures}, and sometimes they are called \emph{applicative structures}.)
 But this is rather an expected result, since one can build a model canonically from the proof system.
This is not what we are doing in this paper.  We are building full models,
and we are studying applicative order terms without variables or abstraction. 
 Our work is thus drastically simpler on the syntactic side, and more complex on the semantic side:
 we  call on and develop results specific 
to preorders (see Sections~\ref{section-background},
\ref{section-complete}, and \ref{section-technical}).  

The main logical system in this paper is presented in Sections~\ref{section-syntax} (the syntax and semantics)
and~\ref{section-proof-system} (the proof system).    Briefly, the syntax allows us to declare that a given function symbol
$f$ be interpreted as a monotone function by writing $f^\uar$.
We also might declare that
$f$ be interpreted as an antitone function by writing $f^\dar$. The basic assertions in the system are inequalities $t\leq u$
between terms of the same type.
The main semantic objects are full structures in the setting of preorders.   
The consequence relation $\Gamma\models t^* \leq u^*$ is defined much as in (\ref{eq-consequence}), except that we use
an order relation in the obvious way, and that we permit a set $\Gamma$ of extra hypotheses.

The main completeness result ought to be a completeness result for a logical system. 
We would like to have 
 $\Gamma\proves t^*\leq u^*$ 
 iff 
  $\Gamma\models t^*\leq u^*$.
 We have a sound logical system; the rules are in Figure~\ref{fig-rules1}. 
We did not obtain a completeness theorem for this system, though we believe it to hold. 
But we do have a related completeness result, Theorem~\ref{theorem-main-wc}.
The formulation restricts the full models
to full models whose base preorders are \emph{weakly complete} (every pair of elements has an upper bound and a lower bound)
and the logic adds a few rules to compensate.
Curiously, there is an echo here from tonoids.  We did not discuss the second requirement on tonoids that the 
underlying poset $\AA$ be a bounded (distributive) lattice and that the 
operation symbols either respect $0$ or $1$.  Every lattice is trivially a weakly complete preorder.

\paragraph{Related work}
We have already mentioned papers on monotonicity calculi.  This paper is the first in the area to present a completeness results for full structures,
the intended semantic models.

The original completeness theorem of Friedman which we mentioned above has been extended
in a few directions.     
Dougherty and Subrahmanyam~\cite{DoughertyS00} extend
the STLC by adding product and coproduct types and a terminal type,
 and they obtain the completeness theorem for full structures.
As far as we know, this is the only extension that obtains completeness for full models on sets.
Several papers move from sets to other categories in order to obtain completeness results,
and the completeness here is the strong completeness theorem $\Gamma\models t = u$ iff  $\Gamma\proves t = u$
which is not available in sets.  For more on this topic, see Awodey~\cite{Awodey00}.

\subsection{Background: preorders and polarized preorders}
\label{section-background}

\begin{definition}
A \emph{preorder} is a pair $\PP=(P,\leq)$, where $P$ is a set, and  $\leq$ is a reflexive and transitive relation on $P$.
Although we technically should use $P$ for the universe of the preorder, we sometimes write $p\in \PP$ when we mean $p\in P$.
If $p\leq q $ and $q\leq p$, then we write $p \equiv q$.   It is possible that $p\equiv q$ without having $p = q$.

Let $\PP$ and $\QQ$ be preorders, and consider a function $f: P\to Q$.
\begin{enumerate}
\item $f$ is \emph{monotone} if whenever $p\leq q$ in $\PP$, $f(p) \leq f(q)$ in $\QQ$.
We also say that $f$ is \emph{order-preserving} in this case.    We write $f^\uar:\PP\to\QQ$.
\item $f$ is  \emph{antitone} if whenever $p\leq q$ in $\PP$, $f(q) \leq f(p)$ in $\QQ$.
We write $f^\dar:\PP\to\QQ$.
\item $f$ is \emph{order-reflecting} whenever $f(p) \leq f(q)$ in $\QQ$, $p\leq q$ in $\PP$.
\item $f$ is an \emph{order embedding} if $f$ is one-to-one, and preserves and reflects the order.
\end{enumerate}
\end{definition}
The logical systems in this paper are about monotone and antitone functions.   But some of the proofs also use the concepts
of order-reflecting functions and order embeddings.

\begin{example}\label{example-up-down}
Here are some examples of the kinds of facts of interest in this paper.
\begin{enumerate}
\item If $f^\uar\leq g^\dar:\PP\to\QQ$, and $a \leq c \geq b$, then $f(a)\leq g(b)$.
This is because $f^\uar(a) \leq f^\uar(c) \leq g^\dar(c) \leq g^\dar(b)$.
\item  If $f^\dar\leq g^\uar:\PP\to\QQ$, and $a \geq c \leq b$,  then again $f(a)\leq g(b)$.
This is similar: $f^\dar(a) \leq f^\dar(c) \leq g^\uar(c) \leq g^\uar(b)$.

\item On the other hand, here 
is an example where $f^\dar \leq g^\uar$, $a \leq c \geq b$, but $f(a) \not\leq g(b)$.
Let $\PP$ be the poset $\set{a,b,c}$ with $a < c >  b$, and let $\QQ$ be  $\set{0,1}$ with $0 < 1$.
Let $f(a) = 1$, $f(b) = 0$, and $f(c) = 0$.   
Let $g(a) = 1$, $g(b) =0$, and $g(c) = 1$.
\label{notsoundalways}
\item  It is possible for a function to be both monotone and antitone.
In our notation, 
it is possible that $f^\uar:\PP\to \QQ$ and also $f^\dar:\PP\to \QQ$.
One way for this to happen is when $f$ is a constant function.
Another way is when $\PP$ is the flat preorder (also called the discrete preorder) on some set $S$: $p_1 \leq p_2$ iff $p_1 = p_2$.
\end{enumerate}
\end{example}

Some additional definitions and constructions concerning preorders appear 
later in this paper, closer to where they are used.   At this point, we introduce
\emph{polarized preorders}, a
type of structure that extends preorders.

\begin{definition}\label{def-ordered-fn-symbols}
A \emph{polarized preorder} is a tuple
 \[\FF= (F,\leq,\uar,\dar),\]
  where $F$ is a set, $\leq$ is a pre-order on $F$, and $\uar$ and $\dar$ are subsets of $F$.
\end{definition}

\begin{example}
\label{example-polarized-preorder}
For any preorders $\PP$ and $\QQ$, we have a polarized preorder $\QQ^{\PP}$ defined as follows.
The set of points of $\QQ^\PP$ is the set $Q^P$ of all functions from $P$ to $Q$.
The order is the pointwise order.   We take $\uar$ to be the set of monotone $f:\PP\to\QQ$,
and $\dar$ to be the set of antitone $f:\PP\to\QQ$.
\end{example}

We also have abstract examples.  
   In such polarized preorders, we think of the 
sets $\uar$ and $\dar$ as providing a specification for what we want them to be in an interpretation.
  Thus we think of them as ``tagged'' $\uar$ or $\dar$ (or possibly neither, or both).
To say that $f$ is tagged $\uar$ just means that $f\in\uar$; similarly for $\dar$.
(A given function symbol might thus be tagged with neither $\uar$ or $\dar$, and it might also be tagged with both
symbols.)
We use $f^\uar$ to range over elements  $f\in F$ which are tagged $\uar$, 
and we also use $f^\dar$ to range over elements  $f\in F$ which are tagged $\dar$.
(And when we write $f$ without $\uar$ or $\dar$, we mean an arbitrary element of $F$.)

\begin{definition} Let $\FF$ be a polarized preorder, and let $\PP$ and $\QQ$ be preorders.
An \emph{interpretation of $\FF$ in $\PP$ and $\QQ$} is a 
 function $\semanticsalt{\ }:\FF \to \QQ^{\PP}$ which is monotone and 
 preserves polarities.  That is, if $f^\uar$, then $\semanticsalt{f}$ is monotone, and
 if $f^\dar$, then $\semanticsalt{f}$ is antitone.
 \label{definition-interpretation}
\end{definition}

This definition will not be used much in this paper, but is shows where things are going.
We think of $\FF$ as ``syntax'' and $ \QQ^{\PP}$ as  the ``semantic space'', and $\semanticsalt{\ }$ as 
the interpretation of the syntax in that space.

\section{Syntax and Semantics}

This section sets the stage for the rest of the paper by presenting the syntax and semantics of our system.

\subsection{Syntax, and semantics in full structures}
\label{section-syntax}

We begin with a set $B$ of \emph{base types}.  We use the letter $\beta$ for these.   We make no assumption on the set $B$, and we also do not vary it in what follows.  
Henceforth we leave $B$ out of our notation.

The full set $\Types$ of types is the smallest set such that every base type $\beta$ belongs to $\Types$,
and if $\sigma$ and $\tau$ belong to $\TT$, then so does $\sigma\to\tau$.  The types which are not base types are called
\emph{function types}. 

\begin{definition}\label{D:syntax}
An \emph{(ordered) signature} is a family 
$(\FF_{\beta})_{\beta}$ of preorders, one for each base type, and a family
$(\FF_{\sigma})_{\sigma}$ of polarized preorders, one for each function type $\sigma\in\Types$.
We form \emph{typed terms} $t: \sigma$ by the following recursion:
\begin{enumerate}
\item If $f\in \FF_{\sigma}$, then $f:\sigma$ is a typed term.
\item If $t:\sigma\to\tau$ and $u:\sigma$ are typed terms, then $tu:\tau$ is a typed term.
\end{enumerate}
When we need notation for a signature, we usually write $\FF = (\FF_{\sigma})_{\sigma}$ and think of these
as polarized, except for the base types.
\end{definition}

We use  notation like $t:\sigma$, $u:\tau$, etc., for typed terms.
Usually we drop the types for readability.
Indeed, we only supply the types to make a point about them.   For example,
in (\ref{eq-semantics}) below, the second equation exhibits the types.
If we were to write $\semantics{tu} = \semantics{t}(\semantics{u})$ without the types,
it could cause a confusion on first reading.
We could use parentheses as well, but these will not be necessary.   When we speak of terms, we usually do not mention 
the underlying signature.

The assertions in the language are inequalities 
$t: \sigma\leq u: \sigma$
between terms of the same type.
Again, we usually drop the types and just write $t\leq u$.

\begin{example}
\label{example-first-signature}
For all relations $R$, we write $R^\star$ for the reflexive and transitive closure of $R$.  So $R^\star$ is the smallest preorder including $R$.

Let  $\beta$ be a base type, let
$\tau$ be any type, so that
 $\beta\to\tau$ is a function type, and let $\FF$ be the signature given by
\[
\begin{array}{lcl}
\FF_{\beta} & = & (\set{a,b,c}, \emptyset^{\star})\\
\FF_{\beta\to\tau} = (F_{\beta\to\tau},\leq, +,-)& = & (\set{f, g}, \set{(f,g)}^{\star}, \set{f}, \set{g})\\
\FF_{\tau\to(\beta\to\tau)} = (F_{\tau\to(\beta\to\tau)},\leq, +,-)& = & (\set{\phi},\emptyset^{\star}, \emptyset,\set{\phi})\\
\end{array}
\] 
For other function types $\mu$, we take $F_{\mu} = (\emptyset, \emptyset, \emptyset,\emptyset)$.
In this signature, we are taking $a$, $b$, and $c$ to be symbols of type $\beta$.  There is no order relation among
these, but our signature does have the reflexivity assertions $a\leq a$, $b\leq b$, $c\leq c$.   
Since $\beta$ is a base type, 
there is no polarization assertion for these symbols.
As for $\beta\to\tau$, we have two symbols $f$ and $g$.   Our signature records $f\leq g$ and that $f\in\uar$ and $g\in\dar$.
When working with this signature, we usually will keep the polarization assertions in mind by repeatedly tagging the symbols.
So we would summarize the polarized preorder $\FF_{\beta\to\tau}$ by simply writing $f^\uar \leq g^\dar$.

Typed terms in our signature include 
$\phi^\dar(f^\uar(a)): \beta\to\tau$.
We could omit the parentheses without risking confusion and also the type;
we then would just write $\phi^\dar f^\uar a$.
An example term of type $\tau$ is $(\phi^\dar f^\uar a)b$. 
\end{example}

\paragraph{Semantics: full structures}

Fix  a 
family of preorders $(\PP_{\beta})_{\beta}$, one for each base type $\beta$.
The family $(\PP_{\beta})_{\beta}$
induces a family  of preorders $(\PP_{\sigma})_{\sigma}$ by
\begin{equation}
\label{fulltypes}
 \PP_{\sigma\arrowdot\tau} = ((\PP_{\tau})^{\PP_{\sigma}}, \leq)
\end{equation}
where $\leq$ is the \emph{pointwise order} on the function set $(P_{\tau})^{P_{\sigma}}$.
For function types $\sigma\to\tau$
we use the polarized preorder structure mentioned in Example~\ref{example-polarized-preorder}:
 for
$f:  P_{\sigma}\to P_{\tau}$, 
we tag $f^\uar$ if $f$ is monotone, and we tag $f^\dar$ if $f$ is antitone.
If $f$ is neither monotone nor antitone, it would be 
 tagged with neither polarity,
If $f$ were both monotone and antitone (see Example~\ref{example-first-signature}(4)), 
it would be tagged both $\uar$ and $\dar$.
What we have built is called the \emph{full preorder type structure} over $(\PP_{\beta})_{\beta}$.

\begin{definition}
\label{def-full-model}
Fix a  signature  $\FF$. 
A \emph{full $\FF$-structure} is a family of preorders
\[ \Model = ((\PP_{\sigma})_{\sigma},  \semantics{\ }) \]
where $(\PP_{\sigma})_{\sigma}$
is the full preorder type structure over $(\PP_{\beta})_{\beta}$
together with 
 a function  $\semantics{\ }$ defined on the typed terms over $\FF$ with the following properties:
 \begin{enumerate}
 \item If $f$ in $\FF_{\sigma}$, then $\semantics{f}\in P_{\sigma}$.
\item For function types $\sigma$, $\semantics{ \ }$ restricts to a map
 $\semantics{\ }_{\sigma}: \FF_{\sigma}\to \PP_{\sigma}$ which is monotone and preserves polarity.
 \end{enumerate}
In other words: if $f\leq g$ in $\FF_{\sigma}$, then $\semantics{f}\leq \semantics{g}$ in $\PP_{\sigma} $;
if $f^{\uar}: \sigma\to\tau$, then 
$\semantics{f}: \PP_{\sigma}\to\PP_{\tau}$ is monotone;  and 
if $f^{\dar}: \sigma\to\tau$, then 
$\semantics{f}: \PP_{\sigma}\to\PP_{\tau}$ is antitone.

Let us emphasize that in a full structure,  (\ref{fulltypes}) holds.
Thus, in a full $\FF$-structure, each function type $\sigma\to\tau$ gives us an interpretation of $\FF_{\sigma\to\tau}$ in $\PP_{\sigma}$ and $\PP_{\tau}$
in the sense of Definition~\ref{definition-interpretation}.  Indeed, 
a full $\FF$-structure amounts to a family of such interpretations together with maps $\FF_{\beta}\to\PP_{\beta}$ for the base types
which preserve the order.

\paragraph{Interpreting typed terms in full structures} 
Fix a full $\FF$-structure $\Model$.  
By recursion on typed terms $t: \sigma$, we define $\semantics{t:\sigma}$:
\begin{equation}
\label{eq-semantics}
\begin{array}{lcl}
\semantics{f:\sigma}  &  &  \mbox{is given in $\Model$, when $f\in \FF_{\sigma}$}\\
\semantics{tu:\tau}  & = & \semantics{t:\sigma\to \tau} (\semantics{u:\sigma})
\end{array}
\end{equation} 
We are using (\ref{fulltypes}) when we see that $\semantics{t:\sigma\to\tau}$ is a function and hence may apply it to $\semantics{u:\sigma}$.
An easy induction shows that when $t:\sigma$, $\semantics{t:\sigma}\in P_{\sigma}$.
As mentioned before, we usually omit the types.
This holds when we use the $\semantics{\ }$ notation.

\end{definition}

\paragraph{Semantic assertions}

Let $t, u$ be terms of the same type $\sigma$.   
We say that $\Model\models t\leq u$ if   $\semantics{t}\leq \semantics{u}$.

Let $\Gamma$ be a set of inequalities $t\leq u$, and let $\Model$ be a full structure.
We say that $\Model\models \Gamma$ if
$\Model\models t\leq u$ whenever $\Gamma$ contains $t\leq u$.
We then speak of a \emph{full model} of $\Gamma$.

Let $\Gamma\cup\set{t^* \leq u^*}$ be a set of inequalities in our language, omitting the types.
We write $\Gamma\models t^* \leq u^*$ if every
full $\FF$-structure which satisfies $\Gamma$ also satisfies $t^* \leq u^*$.
(Incidentally, there is no real reason why we use the $*$ notation on the conclusion $t^*\leq u^*$.
It just permits us to use letters $t$ and $u$ in the rest of an argument, and it also focuses our attention
on one particular assertion of interest.)

In addition, we will need variations on this definition of $\Gamma\models t^*\leq u^*$.
For example, we will be contracting the class of preorders to \emph{weakly complete} preorders (see Section~\ref{section-completion}).
We will
change our notation slightly to clarify the meaning of semantic assertions.
For example
we write  $\Gamma\modelswc A$ if every \emph{weakly complete} model of $\Gamma$ is a model of $A$.

An important point is that our language is built on an ordered signature $\FF$,
and we do not display $\FF$ in our notation $\Gamma\models t^* \leq u^*$.  But this is something to keep in mind.

\begin{example}
Let $\FF$ be as in Example~\ref{example-first-signature}.
Example~\ref{example-up-down} shows that
\[
f^\uar\leq g^\dar, a\leq c, b\leq c \models fa \leq gb.
\]
\end{example}

\begin{remark}
This is perhaps a good place to mention a way in which our overall framework is more permissive than we need it to be.
We allow our signatures to have order assertions $f \leq g$, but all such assertions could be absorbed into a given set $\Gamma$.
So we could have just taken signatures to be families of polarized \emph{sets} rather than polarized preorders.
\end{remark}

\subsection{Proof system}
\label{section-proof-system}

The proof system for the basic logic (without the rules which we shall introduce in Figure~\ref{fig-rules2})
 is shown in Figure~\ref{fig-rules1}.
One point to highlight is that in the (\monorule) and (\antirule) rules,  we have assumptions $f^{\uar}$ and $f^{\dar}$ 
that are part of the underlying signature $\FF$.   There are two ways that we could take these assumptions.  First, we could 
take them to be \emph{side conditions} on the rules.  Doing things that way would mean that we would not show those assumptions
in examples.   The second way would be to take the polarity assumptions to be ``first class''.  This would mean that our proof trees
would not consist solely of inequalities: they could also have assertions from the signature.  This second alternative is the one we 
adopt.  (However, very little would change if we went the other way.) 
In Section~\ref{section-arrow-assertions}, we further 
 extend the proof system in order to 
infer polarity statements about terms; up until then, all polarity assertions in proof
trees occur at the leaves.
With this forward view, we are lead to the formulation which we chose.

\begin{definition}
We write $\Gamma\proves s^*\leq t^*$, where $\Gamma\cup\set{ s^*\leq t^*}$ is a set of assertions in our language, if there is a tree labeled by assertions in
the language whose root is labeled with $s^*\leq t^*$, whose leaves are labeled with elements of $\Gamma$
or with assertions from the underlying signature $\FF$, and such that every non-leaf-node is justified
by one of the rules in  Figure~\ref{fig-rules1}.
\end{definition}

\begin{example} This is a version of Example~\ref{example-up-down}, but done in our proof system.
Let $\FF$ be an ordered signature, and assume that for some type $\sigma\to\tau$, $\FF_{\sigma\to\tau}$ 
contains symbols  $f$ and $g$, and that 
 $f^\uar, g^{\dar}:\sigma\to \tau$.
 
Let $t$, $u$, and $v$ be terms of the same type $\sigma$. Then 
\[
f^\uar\leq g^{\dar}, t\leq v \geq u \proves ft \leq gu.
\]
via the following derivation:
\[
\infer[\transrule]{ft \leq gu}
{
\infer[\transrule]{ft \leq gv}
{
\infer[\monorule]{ft \leq fv}{f^{\uar} & t \leq v}
&
\infer[\pointrule]{fv \leq gv}{f \leq g}
}
&
\infer[\antirule]{gv \leq gu}{g^\dar & u \leq v}
}\]
Observe that the leaves of the tree are assertions in $\FF$.

Similarly, if $f^\dar, g^{\uar}:\sigma\to \tau$, then we have
\[
f^\dar\leq g^{\uar}, t\geq v\leq u \proves ft \leq gu.
\]
This assertion is more naturally made on top of a different ordered signature.
(However, our framework allows the symbols $f$ and $g$ to be declared as both $\uar$ and $\dar$ in a given signature.)
\label{example-key}
\end{example}

\begin{figure}[t]
\begin{mathframe}
\[
\begin{array}{c}
\begin{array}{l@{\qquad}c@{\qquad}l}
\infer[\reflrule]{t \leq t}{}
&
\infer[\transrule]{s \leq u}{s \leq t & t\leq u  }
&
\infer[\pointrule]{su \leq tu}{s \leq t }
\end{array}
\\ \\ 
\begin{array}{l@{\qquad}l}
\infer[\monorule]{ft  \leq fu }{f^{\uar}  & t\leq  u}
&
\infer[\antirule]{fu\leq ft}{f^{\dar} & t \leq u}

\end{array}
\end{array}
\]
\end{mathframe}
\caption{Basic rules of the logic for interpretations in full structures. 
 See Figure~\ref{fig-rules2} for 
additional rules sound for weakly complete preorder structures.\label{fig-rules1}}
\end{figure}

\subsection{The syntactic preorder of a set $\Gamma$, and a construction lemma}

In this section, we fix a  signature  $\FF$ and 
a set $\Gamma$ of inequalities over it. 

\begin{definition}
\label{def-syntactic-preorder}
 For each type $\sigma$, 
$\PPsyn_{\sigma}$ is the set of all terms of type $\sigma$ (not just the constants, the symbols in $F$), the order is provability from $\Gamma$, and $\uar$ and $\dar$
are the constant symbols with the relevant tagging:
\[
\begin{array}{lcl}
\PPsyn_{\sigma} & = & \set{t : \mbox{$t$ is an $\FF$-term of type $\sigma$}}\\
t \leq u & \quadiff & \Gamma\proves t  \leq u\\
t^{\uar} & \mbox{ iff } &  \mbox{$t\in\FF_{\sigma}$, and $t^{\uar}$ in $\FF$} \\
t^{\dar} & \mbox{ iff } & \mbox{$t\in\FF_{\sigma}$, and $t^{\dar}$ in $\FF$}
\end{array}
\]
We call $\PPsyn_{\sigma}$ the \emph{canonical polarized preorder of type $\sigma$}.

Doing this for all $\sigma$ gives
 a family $(\PPsyn_\sigma)_\sigma$ of polarized preorders.
Please note that the family $(\PPsyn_{\sigma})_{\sigma}$ is \emph{not} a full hierarchy over the base preorders.
\end{definition}

\begin{definition}
\label{def-applicative-family}
Let $(\QQ_{\sigma})_{\sigma}$ be a full hierarchy over the  base preorders.
An  \emph{applicative family of interpretations (of $\Gamma)$} is a family
$\Nodel= (\semanticsalt{\ }_{\sigma})_{\sigma}$ of functions indexed by the types
\begin{equation}
\label{eq-applicative-family}
\semanticsalt{\ }_{\sigma}: \PPsyn_{\sigma}\to \QQ_{\sigma}
 \end{equation}
such that each $\semanticsalt{\ }_{\sigma}$ is monotone and preserves polarities on the function types,
and with the following property:  
for all $t \in \PPsyn_{\sigma\to\tau}$ and  $u\in \PPsyn_\sigma$,
\begin{equation}
\label{interlace}
 \semanticsalt{t}_{\sigma\to\tau}(\semanticsalt{u}_{\sigma}) =
  \semanticsalt{tu}_{\tau} .
  \end{equation}
On the left we have function application in the usual sense, and on the right $tu$ is an application on the level of terms.
Please note that an applicative family $\Nodel$ depends on 
a full hierarchy  $(\QQ_{\sigma})_{\sigma}$, and as with everything in this section it 
depends on $\Gamma$ (and thus ultimately on $\FF$).
\end{definition}

\begin{lemma}
\label{lemma-model-construction}
Let $(\QQ_{\sigma})_{\sigma}$ be  a full hierarchy over the base preorders.
Let $\Nodel$ be an 
applicative family of interpretations of $\Gamma$ as in Definition~\ref{def-applicative-family},
 so that (\ref{interlace}) holds.
Then there is a full structure
\[ \Model  = ((\QQ_{\sigma})_{\sigma}, ( \semantics{f})_{f\in \FF}) \]
using the same preorders at each type,
such that the following hold: 
\begin{enumerate}
\item For all $t, u$ of the same type $\sigma$,  $\Model\models t\leq u$ iff in $\Nodel$, $\semanticsalt{t}_{\sigma}\leq \semanticsalt{u}_{\sigma}$.
\label{partoneMC}
\item If each function
$\semanticsalt{\ }_{\sigma}$ preserves the order, then $\Model\models\Gamma$.
\label{part-preserves}
\item If $t^*, u^*:\sigma$, and
$\semanticsalt{\ }_{\sigma}$ reflects the order, and $\Model\models t^*\leq u^*$,
we have $\Gamma\proves t^*\leq u^*$.
\label{part-reflect}
\end{enumerate}
\end{lemma}

\begin{proof}
We define $\semantics{\ }$
by recursion on typed terms (see Definition~\ref{D:syntax}), starting with the case of elements $f\in \FF_{\sigma}$
\[\semantics{f} = \semanticsalt{f}_{\sigma} 
\]
Then we extend to all typed terms by $\semantics{tu:\tau}    = \semantics{t:\sigma\to\tau} (\semantics{u:\sigma}) $.
The difference between $\semantics{\ }$ and $(\semanticsalt{\ }_{\sigma})_{\sigma}$
is that the former is a single function defined on all typed terms by recursion on those terms, while 
$(\semanticsalt{\ }_{\sigma})_{\sigma}$ is a family of functions.
The content of our claim just below is that the two definitions agree. 

\begin{claim} For all $t\in \PPsyn_{\sigma}$,  $\semantics{t} = \semanticsalt{t}_{\sigma}$. 
\label{claim-two-semantics}
\end{claim}

\begin{proof}
By induction on the typed term $t:\sigma$.
The fact that $\semantics{f} = \semanticsalt{f}_{\sigma}$ is immediate for $f\in \FF_{\sigma}$.
Assuming our claim for $t:\sigma\to\tau$ and $u:\sigma$, we see that
\[
\semantics{tu} = \semantics{t}(\semantics{u}) = 
\semanticsalt{t}_{\sigma\to\tau}(\semanticsalt{u}_{\sigma})  = \semanticsalt{tu}_{\tau} 
\]
We used  (\ref{interlace})  at the end.
\end{proof}

This claim easily implies  part (\ref{partoneMC}): 
$\Model\models t\leq u$ iff in $\Nodel$, $\semanticsalt{t}_\sigma\leq \semanticsalt{u}_\sigma$.

For (\ref{part-preserves}),
suppose that $\Gamma$ contains an assertion $t\leq u$.
Let $\sigma$ be the type of these terms.
Then $t\leq u$ in $\PPsyn_{\sigma}$.  Since $\semanticsalt{\ }_{\sigma}$ preserves the order, 
$\semanticsalt{t}_{\sigma}\leq \semanticsalt{u}_{\sigma}$.
By part  (\ref{partoneMC}),
$\Model\models t\leq u$.  


For (\ref{part-reflect}), 
suppose that in $\Model$,  $\semantics{t^*}\leq \semantics{u^*}$.
Then by Claim~\ref{claim-two-semantics}, 
$\semanticsalt{t^*}_{\sigma}\leq \semanticsalt{u^*}_{\sigma}$.
Since  $\semanticsalt{\ }_{\sigma}$ reflects  the order,  
$t^*\leq u^*$ in  $\PPsyn_{\sigma}$.   By the definition of $\PPsyn_{\sigma}$, $\Gamma\proves t^*\leq u^*$.
\end{proof}

The reason that we will be using Lemma~\ref{lemma-model-construction}
in our main result, Theorem~\ref{theorem-main-wc},  is that 
it will be more natural for us to define the functions  $\semanticsalt{\ }_{\sigma}$ by recursion on $\sigma$
than to define $\semantics{ \ }$ by recursion on the typed terms $t$.

\subsection{Digression: incompleteness of the logic with variables on full structures}
\label{section-digression}

 Now that we have the semantics of our inequational typed lambda calculus and also the proof system, we can explain why
this paper is about a logic without variables.   The idea behind our construction comes from an example in
Awodey~\cite{Awodey00} concerning the usual typed lambda calculus: when formulated with variables, it cannot have set theoretic full models.
Take a base type $\beta$ and function symbols $i: (\beta\to\beta)\to \beta$ and $r:\beta$ and the equation $r(i(x)) = x$.
Any full model will interpret the base type $\beta$ by a singleton set.  This leads easily to an incompleteness result for the full semantics in sets.
Although our language does not have the identity symbol $=$, we still get the same result.

In this section, we allow variables and also the rule of substitution: from $t \leq u$, infer $t[s] \leq u[s]$, where $s$ is any substitution.
(That is, any map $s$ which maps variables to terms, respecting the types.)
Let us write $\Gamma\proves t\leq u$ for the proof relation which extends the main proof relation in this paper with this additional rule.

\newcommand{\psidar}{\protect{\psi^{\dar}}}
\newcommand{\phiuar}{\protect{\phi^{\uar}}}

We take one base type, $\beta$,
and symbols $\psidar$, $\phiuar$, $c$, and $d$ with the types shown below:
\[
\xymatrix{    c,d: \beta  \ar@/^.7pc/[rr]^-{\psidar}
 &     &
\ar@/^.7pc/[ll]^-{\phiuar}   \beta\to\beta
}
\]
%
%
For $\Gamma$ we take three inequalities:
\[
\phi(\psi(y)) \leq y \qquad  y \leq \phi(\psi(y)) \qquad \psi(\phi(x)) \leq x
\]
Here we are using a variable $x$  of type $\beta\to\beta$ and a variable $y$ of type $\beta$.

\begin{proposition}
$\Gamma\models c \leq d$, but $\Gamma\not\proves c \leq d$ in the logic using our rules, including substitution.
\end{proposition}

\begin{proof}
Let $\Model$ be a full model of $\Gamma$.   We first observe that for $p,q\in P_{\beta}$, if $p\leq q$, then $q\leq p$.
To see this, write $k$ for $\semantics{\phi}\o \semantics{\psi}$.  So $k$ is antitone, since it is the composition of a monotone and an antitone
function.  Notice that $k(r) \equiv r$ for all $r\in P_{\beta}$,
by our first two assertions in $\Gamma$. Hence if $p\leq q$ then also
$q \equiv k(q) \leq k(p) \equiv p$.

Next, we claim that for all elements  $a, b\in \PP_{\beta}$, $a \equiv b$  in $\PP_{\beta}$.
If not, let $a\not{\!\!\equiv}\, b$.
Define $f: P_{\beta} \to P_{\beta}$ by
\[
f(x)   =   \left\{
\begin{array}{ll} a  & \mbox{if $\semantics{\psi}(x)(x) \equiv b$} \\
b & \mbox{otherwise}
\end{array}
\right.
\]
Since $P_{\beta\to\beta}$ is the full function space, $f$ belongs to it.
Let $x^* = \semantics{\phi}(f)$.
By our last assertion in $\Gamma$,
$\semantics{\psi}(x^*) \leq f$.   
Then 
\[
\semantics{\psi}(x^*)(x^*) \leq  f(x^*)  = 
 \left\{
\begin{array}{ll} a  & \mbox{if $\semantics{\psi}(x^*)(x^*) \equiv b$} \\
b & \mbox{otherwise}
\end{array}
\right.
\]
If $\semantics{\psi}(x^*)(x^*) \equiv b$, then we would also have $\semantics{\psi}(x^*)(x^*) \leq a$. 
But by our first paragraph, we then would have $b \equiv \semantics{\psi}(x^*)(x^*) \equiv a$, and this is a contradiction
to our choice of $a$ and $b$.
So we have $\semantics{\psi}(x^*)(x^*) \not{\!\!\equiv}\,  b$, and thus $\semantics{\psi}(x^*)(x^*) \leq b$. 
Our first paragraph now shows that $\semantics{\psi}(x^*)(x^*) \equiv b$, giving a contradiction again.

It follows from this claim that in our model (hence in any model of $\Gamma$), $\semantics{c} \leq \semantics{d}$.

To complete the proof of our proposition, we make an observation about the particular set $\Gamma$ that we have 
and also our rules, including substitution: if $\Gamma\proves t\leq u$, and if either $t$ or $u$ contains some
given variable of either type
or constant symbol of base type $\beta$, 
then the other term contains it as well.
(For example, we can prove
 $\psi(\phi(x))(y) \leq x(y)$; both sides contain $x$ and $y$.
 We can also prove $\psi(\phi(\psi(d)))(y) \leq \psi(d)(y)$, and both sides contain $d$ and $y$.)
 This observation is proved by an easy induction.
 Thus $\Gamma\not\proves c \leq d$, since $c\leq d$ has $c$ on only one side.
\end{proof}

\subsection{Lemmas on new constants}
\label{section-lemmas-new-constants}

\begin{definition}
\label{definition-extend-F}
Let $\FF = (\FF_{\sigma})_{\sigma}$ be a signature, and write each $\FF_{\sigma}$ as $(F_{\sigma}, \leq, \uar, \dar)$.
For each $\sigma$, let $\Box_{\sigma}\notin F_{\sigma}$.
  Let
\[ \GG_{\sigma} = (F_{\sigma}\cup\set{\Box_{\sigma}}, \leq_{\sigma}^{\star}, \uar, \dar)\]
We have added the new symbol $\Box_{\sigma}$ to $\FF_{\sigma}$.
Notice that the reflexive-transitive closure $\leq_{\sigma}^{\star}$ of $\leq_{\sigma}$ just adds
to $\leq_{\sigma}$ the assertion $\Box_{\sigma} \leq \Box_{\sigma}$,
and $\uar$ and $\dar$ are exactly the same as in $\FF_{\sigma}$.  Thus, 
we add no monotonicity information about the new symbols;
for a type $\sigma\arrowdot\tau$, we do not add to $\Gamma$
assertions like $\Box^{\uar}_{\sigma\arrowdot\tau}$.
$ \GG_{\sigma} $ does not have any ordering relation between any new constant and any other symbol.

This gives a new signature $\GG =(\GG_{\sigma})_{\sigma}$.  Note that we have an inclusion map 
$\iota_{\sigma}:\FF_{\sigma}\to\GG_{\sigma}$ which preserves the order and polarities.
For a set $\Gamma$ of inequalities over $\FF$, we write $\GammaBox$ for the same set, but taking it to be a set of assertions over
 $\GG$.   
 \end{definition}

For a set $\Gamma$ over $\FF$, 
the syntactic and semantic consequence relations $\GammaBox\proves t^*\leq u^*$ 
and   $\GammaBox\models t^*\leq u^*$ are different 
from the ones involving $\Gamma$.
   The main point of the next results is that moving from $\Gamma$ to $\GammaBox$ is a
\emph{conservative extension} in the relevant senses. First, a semantic fact. 

Let $\Model = ((\PP_{\beta})_{\beta}, ( \semantics{\ }_{\sigma})_{\sigma}: \GG_{\sigma}\to \PP_{\sigma})$
 be a full model over $\GG$.  
Let $\Model^0$ be the \emph{reduct to $\FF$}. 
This is  
\[ \Model^0 = ((\PP_{\beta})_{\beta}, ( \semantics{\ }^0_{\sigma})_{\sigma}: \FF_{\sigma}\to \PP_{\sigma})\]
where $\semantics{\ }^0_{\sigma}: \FF_{\sigma} \to \PP_{\sigma}$,
is $ \semantics{\ }_{\sigma}\o \iota_{\sigma}$.

\begin{lemma}
\label{lemma-reduct}
For every inequality $t\leq u$ over $\FF$, $\Model\models t\leq u$ iff $\Model^0\models t\leq u$.
\end{lemma}

\begin{proof}
An easy induction shows that for all terms $t$ over $\FF$, the interpretations of $t$ in $\Model$ and $\Model^0$ are the same:
$\semantics{t} = \semantics{t}^0$.
\end{proof}

\begin{lemma} If $\Gamma \models t^* \leq u^* $ and $\GammaBox$ comes from extending the underlying signature with 
new constants, then 
$\GammaBox \models t^* \leq u^* $.
\label{lemma-no-new-symbols-semantic-second}
\end{lemma}

\begin{proof}
Let $\Model$ be a  $\GG$-structure which satisfies $\GammaBox$.  
By Lemma~\ref{lemma-reduct},
$\Model^0 \models \Gamma$.   By hypothesis,  $\Model^0 \models t^* \leq u^* $.
Then by Lemma~\ref{lemma-reduct} again, $\Model\models  t^*\leq u^*$.
\end{proof}

We now turn to some syntactic results that again point to a conservative extension.

\begin{lemma} If $\GammaBox \proves t \leq \Box$ or $\GammaBox \proves \Box\leq t$, then $t = \Box$.
\label{lemma-Box-first}
\end{lemma}

\begin{proof} By induction on derivations.  With a conclusion like $t\leq \Box$, the  derivation can only use (\reflrule) or (\transrule).
The inductive step for (\transrule) is trivial.
\end{proof}

\begin{lemma}\label{lemma-onebox}
 If $t:\sigma\to\tau$ and $v:\sigma$ and $\GammaBox\proves t\leq v\Box$, then   there is some 
$u: \sigma\to \tau$ 
 so that $t = u \Box$, and  $\GammaBox \proves u\leq v$.
 
 Similarly, if $\GammaBox\proves v\Box\leq t$, then   there is some 
   $u$ 
 so that $t = u \Box$, and  $\GammaBox \proves v\leq u$. 
  \end{lemma}

\begin{proof}
Each part is proved
by induction on the derivation.  The step for (\transrule) is easy. 
 If the root uses (\monorule), or (\antirule), then $t\Box$ is $v\Box$ by Lemma~\ref{lemma-Box-first};
 in this case, $t= v$.
 If it uses (\pointrule), then we directly have that $t\leq v$.
 \end{proof}

\begin{lemma}  
\label{lemma-richness} 
  If $\GammaBox\proves t\Box\leq u\Box$, then $\GammaBox\proves t\leq u$.
\end{lemma}

\begin{proof}
By induction on the derivation.
If the root uses (\pointrule), $t\leq u$.
   If it uses (\reflrule),  (\monorule) or (\antirule), $t$ is $u$.
  Suppose that the root uses (\transrule), say
  \[ \infer[\transrule]{t\Box \leq u\Box}{t\Box \leq v & v\leq u\Box }\]
The previous lemma applies to both subproofs.
  There is some $t\leq x$ so that $v = x \Box$. 
  There is also some $w\leq u$ so that $v = w \Box$. 
So $x = w$.  And then $t\leq w \leq u$ tells us that $t\leq u$.
\end{proof}

\begin{lemma} If $\GammaBox \proves t^* \leq u^* $ with none of the new symbols $\Box_{\sigma}$ occurring in $t^* $ or $u^* $,
then $\Gamma\proves t^*  \leq u^* $.
\label{lemma-no-new-symbols}
\end{lemma}

\begin{proof}
Call a type $\sigma$ \emph{inhabited} (in a given signature)
if there is a term of type $\sigma$ other than $\Box_{\sigma}$.
For each inhabited type, pick a term $t_{\sigma}$ of that type.
Consider the following substitution:
\[
s(\Box_{\sigma}) = \left\{
\begin{array}{ll}
t_\sigma & \mbox{if $\sigma$ is inhabited}\\
\Box_{\sigma} & \mbox{otherwise}
\end{array}
\right.
\]

We claim that if we take any proof tree $\TT$ over $\GammaBox$ and apply this substitution to all terms,
the result $\TT[s]$ is a valid proof tree over $\GammaBox$. 
The proof of this is by induction. 

We next claim that in $\TT[s]$   every assertion $t\leq u$ has the property that some $\Box_{\sigma}$ occurs
in $t$ iff it occurs in $u$.   The proof is by induction, and the main interesting steps  are for (\transrule).

We now fix a proof tree $\TT$ showing that $\GammaBox \proves t^* \leq u^* $.
Now none of the $\Box$ symbols occur in the root $t^*\leq u^*$, or in any of the leaves of the tree.
So the leaves and root of $\TT[s]$ are the same as those of $\TT$.
We claim that 
in $\TT[s]$, 
  every $\sigma$ which occurs  is inhabited.  For this we argue by contradiction; suppose it is false.
    Since the root
has no $\Box$ occurrences, there must be a node in the proof tree  which does have a $\Box$-occurrence
but whose child (downward) in the tree has no $\Box$-occurrences.  The only way this can happen is at the transitivity
step:
\[
\infer{t\leq v}{t\leq u & u  \leq v}
\]
But the observation above applies (twice) and tells us that both $t$ and $v$ have a $\Box$-subterm;
hence $t\leq v$ has at least two of them -- a contradiction! 
Therefore every type in $\TT[s]$ is inhabited.   And then in passing from $\TT$ to $\TT[s]$, we removed $\Box_{\sigma}$ in favor of a term $t_{\sigma}$.
  We conclude that $\TT[s]$ has no $\Box$-terms.
Thus,  $\TT[s]$  is a proof tree over $\Gamma$.   And as we have seen, its leaves and root are the same as those of $\TT$.
\end{proof}

\section{Completeness for Full Weakly Complete Structures in the Extended Logic}
\label{section-completion}

The work in the previous section suggests that
we should prove a completeness theorem for reasoning
in full structures $\Gamma\proves t^*\leq u^*$ iff $\Gamma\models t^*\leq u^*$, where the proof system is the one in Figure~\ref{fig-rules1}
and the semantic notion is based on the full structures which we have introduced in Definition~\ref{def-full-model}.
We have not been able to obtain this result.   On the other hand, we have related results.  First, we might well relax the condition
of fullness to the natural weaker condition associated with Henkin-like models of the typed lambda calculus.   Doing this leads to a
completeness result fairly easily, not just for the logic of this paper but for much more expressive formalisms that have 
a richer type system, variables, abstraction, and arbitrary sets of hypotheses.   This is not the topic of this paper, but for work in this
area, see~\cite{IcardMoss21,IMT,Tune2016}.  (We should mention that \cite{IMT} has an error that will be fixed in a follow-up publication.)

\begin{definition}
\label{def-wc}
A preorder is \emph{weakly complete} if every $x$ and $y$ have some upper bound $z$
and also some lower bound $w$.  The bounds required need not be least upper bounds or greatest lower bounds.
A full structure is called \emph{weakly complete} if every base preorder $\PP_{\beta}$ is weakly complete.
(It follows that each $\PP_{\sigma}$ is weakly complete.)
\end{definition}

As the name suggests, weak completeness is a fairly weak property.   Every lattice has this property, for example.
Every preorder with a greatest and a least element is weakly complete.
On the other hand,
a flat preorder containing two or more points is not weakly complete.
A disjoint union of two non-empty preorders is also not weakly complete.

\begin{figure}[t]
\begin{mathframe}
\[
\begin{array}{c}
\infer[\wcintrees_1]{ft \leq gu}{f^{\uar} : \sigma\to\tau & g^{\dar}: \sigma\to\tau & f\leq g  & }
\qquad
\infer[\wcintrees_2]{ft\leq gu}{f^{\dar} : \sigma\to\tau & g^{\uar}: \sigma\to\tau & f\leq g  & }
\\ \\
\infer[\wcintrees_3]{g \leq h}
{ f^{\dar},g^{\uar}, h^{\uar}, k^{\dar} : \sigma\to\tau  & g\leq f  & k \leq h & ft \leq ku}
\end{array}
\]
\end{mathframe}
\caption{Additional rules of the logic which are sound  for weakly complete preorders. 
(\wcthree) 
stands in for  four rules; we could also have the following arrangements at the front:
(a) $ f^{\uar},g^{\dar}, h^{\uar}, k^{\dar}$;
(b) $ f^{\dar},g^{\uar}, h^{\dar}, k^{\uar}$;
(c)  $f^{\uar},g^{\dar}, h^{\dar}, k^{\uar}$.
\label{fig-rules2}}
\end{figure}

The logic relevant to weakly complete full structures is given in Figure~\ref{fig-rules2}, 
taken in addition to the rules which we saw in Figure~\ref{fig-rules1}.

Suppose that $f$ and $g$ are function symbols of the same type, say $\sigma$.
We write $f\lequardar g$ to mean that either 
$ f^{\uar} \leq g^{\dar}$ 
or else that 
$ f^{\dar} \leq g^{\uar}$.
With this notation, the six (\wcrule) rules maybe written as two:
\[
\infer[\mbox{\sc wc${}_{1,2}$}]{fa \leq gb}{f \lequardar g  &a \leq b}
\qquad\qquad
\infer[\mbox{\sc wc}_3]{g \leq h}
{ g\lequardar f  & k \lequardar h  &  ft\leq ku  }
\]

We write $\Gamma\proveswc  s^*\leq t^*$ if there is a derivation (a proof tree)
  that also allows the weak completeness rules in Figure~\ref{fig-rules2}.
 And $\Gamma\modelswc  s^*\leq t^*$ means that every weakly complete full model of $\Gamma$ is also a model of $ s^*\leq t^*$.

\begin{proposition}
 If $\Gamma\proveswc  s^*\leq t^*$, then $\Gamma\modelswc  s^*\leq t^*$.
\label{prop-soundness-wc}
\end{proposition}

\begin{proof}
By induction on proofs in the system. 
We only consider the (\wcrule) rules. 
For (\wcone), fix a 
weakly complete full structure $\Model$.
We know that $\semantics{f}: \PP_{\sigma}\to \PP_{\tau}$ is a monotone function,
$\semantics{g}: \PP_{\sigma}\to \PP_{\tau}$ is an antitone function,
and also $\semantics{t}, \semantics{u}\in P_{\sigma}$.
By weak completeness of $\PP_{\sigma}$,
let $x\in P_{\sigma}$ be such that $\semantics{t}, \semantics{u}\leq x$.
Then by Example~\ref{example-up-down}, $ \semantics{f}(\semantics{t}) \leq  \semantics{g}(\semantics{u})$.
Thus $\semantics{ft}\leq \semantics{gu}$.

The soundness of (\wctwo) is similar, and it uses the fact that every pair of elements of $P_{\sigma}$ have some lower bound.

Next, let us consider  (\wcthree) with the same notation as just above.
The important thing is that the premises do not include $f\leq k$, just the much weaker assertion that for particular terms $t$ and $u$,
$ft \leq ku$.
But this is enough: take any $x\in P_{\sigma}$ and observe
\[
\begin{array}{lcll}
\semantics{g}(x) & \leq & \semantics{f}(\semantics{t}) & \mbox{by Example~\ref{example-up-down}} \\
& = & \semantics{ft} & \mbox{by the recursive clauses in the semantics}\\
& \leq & \semantics{ku}   & \mbox{by the overall induction hypothesis, and $ft\leq ku$} \\
& = & \semantics{k}(\semantics{u})\\
& \leq & \semantics{h}(x) & 
\end{array}
\]
Since $x$ was arbitrary, we have shown that $\semantics{g}\leq \semantics{h}$ pointwise.
\end{proof}

The calculation just above makes it clear that the last two premises could be changed.    For example, we could 
have $g^{\dar}$, $h^{\dar}$, $f^{\uar}$, and $k^{\uar}$.
The only thing that matters is that the arrow directions $g$ and $f$ have to be opposite, and the same goes for $h$ and $k$.
So there are four (\wcthree) rules.

\subsection{Additional lemmas on new constants}
\label{section-lemmas-new-constants-weakly-complete}

We proved results in 
Section~\ref{section-lemmas-new-constants} that showed how adding fresh constants to a signature
gives a 
conservative extension both for the semantics and the proof theory.
At this point, we need to re-work that section
in light of the new (\wcrule) rules.  
Definition~\ref{definition-extend-F} mentioned notation having to do with new constants.  This needs no change.
The semantic results in 
Lemma~\ref{lemma-reduct} and~\ref{lemma-no-new-symbols-semantic-second} do not change:
the reduct of a weakly complete model is weakly complete.
No change is needed in Lemma~\ref{lemma-Box-first}, since none of the (\wcrule) rules allow us to conclude an inequality 
whose left- or right-hand side is a new symbol $\Box_{\sigma}$ by itself.
Lemma~\ref{lemma-onebox} does need to change.

\begin{lemma}\label{lemma-onebox-wc}
 If $\GammaBox\proveswc t\leq v\Box$, then one of the following holds: 
\begin{enumerate}
    \item 
   There is some $u\leq v$ so that $t = u \Box$.  
   \item There is a term  $s:\sigma$, 
   and constants $f,g:\sigma\to\tau$ such that $\GammaBox\proveswc t\leq f s$, and
   $f\lequardar  g\leq v$.  \end{enumerate}
\end{lemma}

\begin{proof}
By induction on the number of ({\sc trans}) steps in the derivation.
We cannot have a derivation where the root is  $t\leq u\Box$ justified by (\wcthree).
   Applications of (\wcthree) conclude an inequation between function symbols which have a declared $\uar$ or $\dar$ marking.
\end{proof}

We also have a parallel result for the situation  $\GammaBox\proves u\Box \leq t$.

\begin{lemma}  
\label{lemma-richness-wc} 
  If $\GammaBox\proveswc t\Box\leq u\Box$, then $\GammaBox\proveswc t\leq u$.
\end{lemma}

\begin{proof}
By induction on the the height of the   derivation.    If the root is (\reflrule),  (\monorule) or (\antirule), $t$ is $u$.
 If the root is (\pointrule), we see that $t\leq u$.
If the root is (\wcone) or (\wctwo),  then we have $t \lequardar u$.  In particular,  $t\leq y$.
As in Lemma~\ref{lemma-onebox-wc}, we cannot have a derivation where the root is (\wcthree) and where the assertion at the root is  $t\Box\leq u\Box$.

The main work is when the root is (\transrule), say
  \begin{equation}
  \label{boxcases}
   \infer[\transrule]{t\Box \leq u\Box}{t\Box \leq v & v\leq u\Box}
   \end{equation}
The first case is when we have 
   two instances of the first option in Lemma~\ref{lemma-onebox-wc}.
   The proof works as in  Lemma~\ref{lemma-onebox}.

Suppose first that we have  the first option in Lemma~\ref{lemma-onebox-wc} on the left  premise 
of (\ref{boxcases}), 
 say with $v$ being $w\Box$ and  $t\leq w$.
 Then the right premise above is $w\Box \leq u\Box$.
 By induction hypothesis $\GammaBox\proves w \leq u$.
 But then using (\transrule) we have $t \leq u$, as desired.
  
 The same  reasoning applies if the first option in Lemma~\ref{lemma-onebox-wc} 
 applied to the right premise  of (\ref{boxcases}).
 
The most interesting case is when 
both premises of (\ref{boxcases}) give instances of the 
 second option in Lemma~\ref{lemma-onebox-wc}.
From the left premise  $t\Box \leq v$,  we get
   $x$,  $y$, and $z$ such that $ xy\leq v$, and 
   $t \leq z\lequardar x $. 
From the right premise $v\leq u\Box$ we get
 $f$,  $d$, and $e$ such that $ v\leq f d$, and
   $f\lequardar  e\leq u$. 
Then    $xy\leq v \leq fd$, and also  $z\lequardar x $ and   $f\lequardar  e$.
From (\wcthree), we get $z\leq e$.
By this fact together with $t\leq z$ and $e\leq u$, we have $t\leq u$.
\end{proof}

\begin{lemma} If $\GammaBox \proveswc t^* \leq u^* $ with none of the new symbols $\Box_{\sigma}$ occurring in $t^* $ or $u^* $,
then $\Gamma\proveswc t^*  \leq u^* $.
\label{lemma-no-new-symbols-wc}
\end{lemma}
\begin{proof}
The proof of this result elaborates the proof of Lemma~\ref{lemma-no-new-symbols}.
We begin again with the observation that 
 if we take any proof tree $\TT$ over $\GammaBox$ in this system and replace, for every inhabited type $\sigma$, 
  every occurrence of $\Box_{\sigma}$ by a fixed term $t:\sigma$
which is not $\Box_{\sigma}$,
the result is a valid proof tree $\TT[s]$  over $\GammaBox$.

 We also claim that in $\TT[s]$ every assertion $t\leq u$ has the following property: for all types $\sigma$,   $\Box_{\sigma}$ occurs
in $t$ iff it occurs in $u$.  
In the induction this time, we do not have to worry about (\wcthree), since conclusions of (\wcthree) cannot involve a new symbol.
But we do need to think about (\wcone) and (\wctwo).  It allows us to conclude an inequality $fx \leq gy$ where $x$ and $y$ are 
possibly new (not in the subterm above) terms of the same type $\sigma$.   Indeed, $x$ and $y$ might possibly be  $\Box$.
If both or neither is $\Box$, then we are done.  And we cannot have one being $\Box$ and the other not, since this would imply that 
$\sigma$ is inhabited and that the $\Box$-occurrence would have been replaced.

The end of the proof expands on that of  
 Lemma~\ref{lemma-no-new-symbols}.
In the other proof, we took a proof tree $\TT$ showing that
$\GammaBox \proveswc t^* \leq u^* $ and directly showed that $\TT[s]$ could have no $\Box$-occurrences.
This time we might have applications of (\wcthree) that get rid of two $\Box$-occurrences as we go from top to bottom
on the left below:
\[\infer[\wcintrees_3]{g \leq h}
{ \infer*{g\lequardar f}{}  & \infer*{k \lequardar h}{}  & \infer*{f\Box_{\sigma} \leq k\Box_{\sigma}}{} }
\qquad
\infer[\transrule]{g \leq h}{ \infer[\transrule]{g \leq k}{ \infer*{g \leq f}{} &  \infer*{f \leq k}{}}
& 
 \infer*{k\leq h}{}}
\]
However,  in view of  Lemma~\ref{lemma-richness-wc}, 
the third premise implies that $\GammaBox \proves f \leq k$.
We can thus replace the entire application of (\wcthree) above by two applications of (\transrule) 
in order to conclude that $\GammaBox \proves g \leq h$, as on the right above. 
We do this replacement for every  application of (\wcthree) that dropped two $\Box$-occurrences.  After that, the same proof by induction 
as in  Lemma~\ref{lemma-no-new-symbols} shows that  $\TT[s]$  has no $\Box$-occurrences.  This completes the proof.
\end{proof}

\subsection{Complete preorders}
\label{section-complete}

Our completeness theorem is for full structures which use weakly complete preorders for every type.
But the proof uses the stronger notion of a complete preorder.  The high-level reason is that 
in building a model of some set $\Gamma$ of assumptions, it is very useful to define functions using
joins of sets of elements.   The kind of definition we have in mind would not work out in general on weakly complete preorders.
The pleasant fact is that every preorder has an order embedding into some complete preorder.  Our eventual proof 
strategy will involve taking the syntactic preorders for the base types
$\PPsyn_{\beta}$ determined by $\Gamma$, choosing completions for them, and then building the full hierarchy over the completions.

\begin{definition} A \emph{complete preorder} is a preorder  $\PP$ with the property that every subset $S\subseteq P$ has a 
least upper bound.  This is an element $\bigvee S\in P$ with the property that 
for all $x\in S$, $x\leq \bigvee S$; and if $y$ is such that 
for all $x\in S$, $x\leq y$, then $\bigvee S \leq y$.

The least upper bound of a set $S$ is not in general unique, but any two least upper bounds $x$ and $y$
have the property that $x\equiv y$.
\end{definition}

In a complete preorder we can fix an operation $\bigvee$ on subsets which gives the least upper bound.
This uses the Axiom of Choice.  Our definition does not build in $\bigvee$ as part of the structure of a complete preorder.
That is, we did not take a complete preorder to be a structure $(P,\leq,\bigvee)$.  But nothing much would change if we had done so.

Notice that if $\PP$ is a complete preorder then 
$\bigvee \emptyset \leq x$ for all $x$, and $x \leq \bigvee P$.
So $\bigvee P$ is  a ``top''.  Similarly  $\bigvee \emptyset$ is a ``bottom.''
In particular, every complete preorder is weakly complete.

\begin{proposition}
If $X$ is a set and 
$\LL = (L,\leq)$ is a  complete preorder, then for all sets $X$, the function set $L^X$ is a 
 complete preorder under the pointwise $\leq$ relation.
 To see this, fix  an $\bigvee$ operation for the subsets of $L$.
For 
  $S\subseteq L^X$, 
 we define 
 $$\begin{array}{lcl}
 (\bigvee S)x & = & \bigvee (\set{f(x): f\in S})
 \end{array}$$
 Then it is easy to see that  $\bigvee$ turns $L^X$ into a complete preorder.
 \label{prop-complete-functions}
\end{proposition}

\begin{remark}\label{worthkeeping}
 Here are two facts worth keeping in mind.
 \begin{enumerate}
\item For sets $A_1,\ldots, A_k$ and $B_1, \ldots, B_{\ell}$ of subsets of $M$, 
\[
\bigvee (A_1 \cup A_2 \cup\cdots A_k) \leq \bigvee (B_1  \cup B_2\cup \cdots\cup B_{\ell})
\]
provided that every $A_i$ is a subset of some $B_j$.  (This sufficient condition is not necessary, but it is sufficient and useful.)
\item Thus, for sets $A, B\subseteq M$, $\bigvee A \leq \bigvee B$ provided that every $a\in A$ is $\leq$ some $b\in B$. 

\end{enumerate}
\end{remark}


\begin{proposition}
Let $\PP = (P,\leq)$ be a preorder.  Then there is a  complete preorder $\PP^* = (P^*,\leq)$ and an order embedding
 $i: \PP \to \PP^*$.
 \label{prop-CPL}
\end{proposition}

\begin{remark} Before we turn to the proof, let us make two comments.
First, we are not claiming any uniqueness of 
$\PP^*$ of $i$ in Proposition~\ref{prop-CPL}.
There are in fact many ways to take a preorder and complete it in our sense.

Second,
for $\PP$ a poset (that is, a preorder additionally satisfying antisymmetry), 
 we may use the usual construction of a complete lattice extending $\PP$ by taking down-closed sets.
However, we need a construction in which distinct elements $p,q\in P$ which are equivalent
($p\leq q\leq p$) are \emph{not} identified by $i$.   So the construction using down-closed sets will not work.
However, it will be close.
We are going to take  the product of the complete lattice of down-closed subsets of $P$ 
by the indiscrete preorder on the set $P$.
\end{remark}

\begin{proof}
 We define the preorder $(P^*,\leq)$ and the map $i$ by 
\[ \begin{array}{lcl}
P^* & = & \set{(A,p) : A\subseteq P \mbox{ is down-closed in $\leq$ and $p\in P$}} \cup\set{\bot}\\
(A, p) \leq (B,q) & \mbox{iff } & A\subseteq B \\
\bot \leq x & & \mbox{for all } x\in P^*\\
i(p) & = & (\set{q\in P: q\leq p}, p) \\
\end{array}
\]
The symbol $\bot$ in $P^*$ is just intended to be some object which is fresh: it should not be a down-closed subset of $P$.
In the definition of $(A, p) \leq (B,q)$, $p$ and $q$ play no role.
To prove that every subset has a least upper bound, 
we need some extra machinery and a piece of notation.
Fix a choice function
\[
\varepsilon: \mathcal{P}(P)\setminus\set{\emptyset}\to P
\]
such that $\varepsilon(W) \in W$ for all nonempty subsets $W\subseteq P$.
For a set $S\subseteq P^*$, define $W = W_{S}$ by
\[
W =  \set{q\in P : \mbox{ for some $(A,p)\in S$, $q\in A$}}\]
Then for each $S\subseteq P^*$ define
\[ \bigvee S = \left\{ 
\begin{array}{ll} (W,\varepsilon(W)) & \mbox{if $W \neq \emptyset$} \\
 \ \bot & \mbox{if $W = \emptyset$}\\
\end{array}
\right.
\]
The reason that we need $\varepsilon$ is that we could take $\bigvee S$ to be $(W,p_0)$ whenever $W$ is non-empty and $p_0\in W$.
All such elements $(W,p_0)$  will be equivalent in $\PP^*$.

It is easy to check that $\PP^*$ is a preorder, and that for all $S\subseteq P^*$, $\bigvee S$ is a least upper bound of $S$.
Here is the verification of
the required properties of $i$.
First, if $i(p) = i(q)$, then by considering the second components of $i(p)$ and $i(q)$, we see that $p = q$.
Continuing, 
if $p \leq q$, then every $r \leq p$ is also $\leq q$, and so
 \[ i(p) = (\set{r : r \leq p}, p)  \leq   (\set{r : r \leq q}, q) = i(q).\]
Conversely, if $i(p) \leq i(q)$, then since $p$ belongs to the first component of
$i(p)$, we see that $p \leq q$.
\end{proof}

\subsection{The Extension Lemma}
\label{section-technical}

We are going to use a technical lemma which allows us to 
take preorders $\MM$ and $\LL$ and to
define a map
 $ \FF \to \MM^{\LL}$
 from a map  $ \FF \to \MM^{\SS}$, where $\SS$ is a ``sub-preorder'' of $\LL$.  
  The work in this section will surely seem unmotivated at first glance.  In fact, it will play 
a key role in our proof of the completeness theorem for the $(\wcrule)$-deductive system.
The reason for separating out this lemma and presenting it here is 
that it will be used infinitely many times 
as part of an inductive construction (see Lemma~\ref{lemma-semanticsalt-wc}).
The reader may wish to omit the proof of  Lemma~\ref{lemma-extend-gpa} on first reading.

\begin{lemma} [Extension Lemma] \label{lemma-extend-gpa} 
Let $\FF$ be a polarized preorder.
Let $\LL$,  $\MM$, and $\SS$  be preorders with $\MM$ complete. 
Let  $j: \SS\to\LL$ be an order embedding.
Let  $p: \FF \to \MM^{\SS}$ preserve the order and polarity, and write $p_f$ for $p(f):\SS\to\MM$.
Assume the following \emph{weak-completeness-like} property: 
\begin{equation}\label{special}
\mbox{whenever $f \lequardar g$ in $\FF$, and $x,y\in S$,  then
$p_f(x) \leq p_g(y)$.}
\end{equation}
Then $p$ has an extension $q: \FF \to \MM^{\LL}$:
$q$ preserves the order and polarity, and 
for all $f\in \FF$, $q_f \o j = p_f$:
\[
 \begin{tikzcd}
 \SS \arrow{r}{p_f}  \ar{d}[swap]{j} & \MM\\
 \LL \arrow[dashrightarrow]{ur}[swap]{q_f} \\
  \end{tikzcd}
  \]
\end{lemma}
\begin{proof}
For each $f\in\FF$ and $x\in L$,  define the following four subsets of $M$:
$$\begin{array}{lcl}
A(f,x) & = &   \set{p_{h^\uar}(s) : h^\uar\leq f, j(s)\leq x, \mbox{ and } s\in S}   \\
B(f,x) & = &    \set{p_{h^\dar}(s):  h^\dar\leq f, x\leq j(s),\mbox{ and } s\in S}  \\
C(f) & = &\set{p_{h^\dar}(s) : (\exists k^\uar \leq f)( h^\dar \leq k^\uar), \mbox{ and } s\in S}   \\
D(f) & = &  \set{p_{h^\uar}(s) :  (\exists k^\dar \leq f) (h^\uar \leq k^\dar), \mbox{ and } s\in S}     \\
\end{array}
$$
For each $f\in\FF$ and $x\in L$ we then define $\Qq f(x)\in M$
by
\begin{equation}
\Qq f(x)  =
  \left\{
\begin{array}{ll}
p_f(s) &\mbox{if for some (unique) $s\in S$,  $x= j(s)$}\\
\bigvee \biggl(A(f,x)    \cup   B(f,x)    \cup   C(f) \cup   D(f)\biggr) &\mbox{if  $x\notin j[S]$}\\
\end{array}
\right.
\label{val}
\end{equation}
Here and also below, we use the fact that if $x = j(s)$, then $s$ is unique.  
This is because $j$ is an order-embedding, hence it is one-to-one by definition.
The join in (\ref{val}) exists because $\MM$ is a complete preorder.

\begin{claim} 
If $x = j(s)$, then every element of $A(f,x)    \cup   B(f,x)    \cup   C(f) \cup   D(f)$ is $\leq p_f(s)$.
\label{overload}
\end{claim}

\begin{proof}

Take an element of $A(f,j(s))$, say  $p_{h^\uar}(t)$ where $j(t)\leq j(s)$.
Since $j$ reflects order, $t\leq s$.
Then   $p_{h^\uar}(t)\leq p_{h^\uar}(s) \leq p_f(s)$.
At the end we used the assumption that $p$ preserves order and polarity: since $h^\uar \leq f$, $p_{h^{\uar}} \leq p_f$ in $\MM^{\SS}$ 
and  $p_{h^{\uar}}$ is monotone.

This time, take an element of $B(f,j(s))$, say  $p_{h^\dar}(t)$ where $j(s)\leq j(t)$.
Since $j$ reflects order, $s\leq t$.
Then   $p_{h^\dar}(t)\leq p_{h^\dar}(s) \leq p_f(s)$.


We turn to $C(f)$. Let  $h^\dar \leq k^\uar \leq f$ in $\FF$ and let $t\in S$.
We have $p_{h^\dar}(t)\leq p_{k^\uar}(s)\leq p_{f}(s)$.


Finally, for $D(f)$, let  $h^\uar \leq k^\dar \leq f$ and $t\in S$.
Then $p_{h^\uar}(t)\leq p_{k^\dar}(s)\leq p_{f}(s)$.



\noindent 
Please note that the points about $C(f)$ and $D(f)$ used the weak-completeness-like property (\ref{special}).
\end{proof}

\begin{claim}
Suppose that $f\leq g$ in $\FF$.  Then for all $x\in \LL$,
$A(f,x) \subseteq A(g,x)$, 
$B(f,x) \subseteq B(g,x)$, 
$C(f) \subseteq C(g)$, 
and $D(f) \subseteq D(g)$.
\label{claimfg}
\end{claim}

\begin{proof}
All parts of this claim are consequences of the transitivity 
of $\leq$ in $\FF$.
\end{proof}

\begin{claim} Suppose that $x\leq y$ in $L$.
Then $A(f,x) \subseteq A(f,y)$, and $B(f,y)\subseteq B(f,x)$ [sic].  
\label{claimfxy}
\end{claim}

\begin{proof}
These are consequences of the transitivity of $\leq$ in $\LL$.
\end{proof}

In the next few claims, we show that 
 $\Qq{f}(x) \leq \Qq{g}(y)$ by showing  that every set involved in the definition of
  $\Qq{f}(x)$ in (\ref{val}) is a subset of some set involved in the definition of $\Qq{g}(y)$.
  This comes from Remark~\ref{worthkeeping}.

\begin{claim}
If $f\leq g$, then $\Qq{f}(x) \leq \Qq{g}(x)$ for all $x\in L$.  Thus, $q: \FF \to \MM^{\LL}$ is monotone.
\label{monclaim}
\end{claim}

\begin{proof} 
If $x\in j[S]$,  say $x = j(s)$, then $q_f(x) = p_f(s) \leq p_g(s) = q_g(x)$.  If $x\notin S$, we see from Claim~\ref{claimfg} that
each of the sets involved in $\Qq{f}(x)$ is a  
subset of the corresponding set involved in  $\Qq{g}(x)$.
So  $\Qq{f}(x) \leq \Qq{g}(x)$.
 \end{proof}

\begin{claim} If $f^\uar$, then $q_{f^\uar}$ is monotone.
\label{taguar}
\end{claim}

\begin{proof}
Let $x\leq y$. We show that $q_{f^\uar}(x)\leq q_{f^\uar}(y)$.
If $x\in j[S]$,  say $x = j(s)$, then $q_{f^\uar}(x) = p_{f^\uar}(s) \in A(f^\uar,y)$.  So $p_{f^\uar}(s) \leq \bigvee A(f^\uar,y) \leq q_{f^\uar}(y)$.

If $x\notin j[S]$, we show that $B(f^\uar,x)\subseteq  C(f^\uar)$.
For then, by
 Claims~\ref{claimfg} and~\ref{claimfxy}, we would have the desired inequality $q_{f^\uar}(x) \leq q_{f^\uar}(y)$.
 In more detail, we would have 
$A(f^\uar,x) \subseteq A(f^\uar,y)$, $B(f^\uar,x)\subseteq  C(f^\uar)$, and obviously 
$C(f^\uar) \subseteq C(f^\uar)$ and $D(f^\uar) \subseteq D(f^\uar)$.
Let $p_{h^\dar}(s)\in B(f^\uar,x)$, where $h^\dar \leq f^\uar$ in $\FF$ and $s\in S$. (We also have $x \leq j(s)$, but this is not used.)
Then
$p_{h^\dar}(s)\in C(f^\uar)$: take $k^\uar = f^\uar$ in the definition of $C(f^\uar)$.
\end{proof}

\begin{claim}  If $f^\dar$, then $q_{f^\dar}$ is antitone.  \label{tagdar}
\end{claim}

\begin{proof}
Let $x\leq y$.  We show that $q_{f^\dar}(y)\leq q_{f^\dar}(x)$.
If $y\in j[S]$, say $y= j(s)$, then $q_{f^\dar}(y) = p_{f^\dar}(s)\in B(f^\dar,x)$.  So $p_{f^\dar}(s) \leq \bigvee B(f^\dar,x) \leq q_{f^\dar}(x)$.


If $y\notin j[S]$, we show that $A(f^\dar,y)\subseteq  D(f^\dar)$.
For then, by
 Claims~\ref{claimfg} and~\ref{claimfxy}, we would have the desired inequality $q_{f^\dar}(y) \leq q_{f^\dar}(x)$.
Let $p_{h^\uar(s)}\in A(f^\dar,y)$, where $h^\uar \leq f^\dar$ in $\FF$, $s\in S$, and $j(s)\leq y$.
Then
$p_{h^\uar}(s)\in D(f^\dar)$: take $k^\dar = f^\dar$ in the definition of $D(f^\dar)$.
\end{proof}

We
 complete the proof of Lemma~\ref{lemma-extend-gpa}.
 We began with $p: \FF\to \MM^{\SS}$
 and defined $q: \FF\to \MM^{\LL}$.
  The verifications that $q$ is monotone and preserves polarity come from 
  Claims~\ref{monclaim}--\ref{tagdar}.
For all $f\in F$,  (\ref{val}) tells us that  $q_f \o j = p_f$.
This completes the proof.
\end{proof}

\subsection{Completeness theorem}

 \begin{theorem} [Completeness]
  If $\Gamma\modelswc  s^*\leq t^*$, then
 $\Gamma\proveswc  s^*\leq t^*$.
 \label{theorem-main-wc}
 \end{theorem}

\begin{proof}
Fix a set $\Gamma$ of inequalities over some signature $\FF$.  Let $\GG$ come from $\FF$ by  adding a fresh constant $\Box_{\sigma}$
of every type $\sigma$.
Let $\GammaBox$ be $\Gamma$, taken as a set of inequalities over $\GG$.  
Let $\PPsyn_{\sigma}$ be as in  Definition~\ref{def-syntactic-preorder}, using $\GG$ and $\GammaBox$.
For each base type $\beta$,  use Proposition~\ref{prop-CPL} to choose a complete preorder
$\QQ_{\beta}$ and an order embedding
$i_\beta: \PPsyn_{\beta} \to \QQ_{\beta}$.
Let the preorders $\PPsem_{\sigma}$ be as defined below:
\begin{equation}
\label{defQ}
\begin{array}{lcl}
\PPsem_{\beta} & = & \QQ_{\beta} \mbox{ from just above}\\
\PPsem_{\sigma\to\tau} & = & \mbox{the full function set $(\PPsem_{\tau})^{\PPsem_{\sigma}}$, ordered pointwise}\\
\end{array}
\end{equation}
On the function types $\sigma$, we construe $\PP_{\sigma}$ as a polarized preorder in the obvious way.
By Proposition~\ref{prop-complete-functions},
each preorder $\PPsem_{\sigma}$  is  complete.
The family $(\PPsem_{\sigma})_{\sigma}$ is a full hierarchy.

In the lemma below, recall the notion of an applicative family of interpretations.  We 
construct such a family using our signature $\FF$ and the full hierarchy  $(\PPsem_{\sigma})_{\sigma}$.

\begin{lemma}
\label{lemma-semanticsalt-wc}
There is an applicative family of interpretations $\Nodel = (\semanticsalt{\ }_{\sigma})_{\sigma}$, where
\[ \semanticsalt{\ }_{\sigma}: \PPsyn_{\sigma}\to \PPsem_{\sigma},\]
such that for base types $\beta$,  $\semanticsalt{\ }_{\beta} =i_{\beta}$,
and for all $\sigma$, 
$\semanticsalt{\ }_{\sigma}$ is an order embedding.
\end{lemma}

\begin{proof}
We define $\PP_{\sigma}$  and $\semanticsalt{\ }_{\sigma}$ by recursion on the type $\sigma$.
We verify that $\semanticsalt{\ }_{\sigma}$ is an order embedding and also for function types that the
relevant applicative family property (\ref{eq-applicative-family}) holds.

The recursion begins with base types. The order embedding fact is stated in Proposition~\ref{prop-CPL},
and there is nothing to check concerning the applicative family property.
 
In the induction step,
we assume that  $\semanticsalt{\ }_{\sigma}$ and  $\semanticsalt{\ }_{\tau}$   
are order embeddings.
We shall  define $\semanticsalt{\ }_{\sigma\to\tau}$ 
 using 
 Lemma~\ref{lemma-extend-gpa}.
The role of $\FF$ in the lemma will be played by 
the polarized preorder  $\PPsyn_{\sigma\to\tau}$;
 please note that we are not using the preorder given by the original signature
 but by its closure under the logic.
 We further take
 $\LL= \PPsem_{\sigma}$,
$\MM = \PPsem_{\tau}$,
$\SS = \PPsyn_{\sigma}$,  $j: \SS\to \LL$ to be $\semanticsalt{\ }_{\sigma}$,
and $p: \FF\to  \MM^{\SS}$ to be given by $p_t(u) = \semanticsalt{t u }_{\tau}$.
In pictures, here is what is going on.  For each term $t:\sigma\to \tau$, we obtain $\semanticsalt{t}_{\sigma\to\tau}$
as shown below:
\[
 \begin{tikzcd}[column sep=.75in]
\PPsyn_{\sigma} \arrow{r}{u\mapsto \semanticsalt{tu}_{\tau}}  \ar{d}[swap]{\semanticsalt{\ }_{\sigma}} &  \PPsem_{\tau}\\
 \PPsem_{\sigma} \arrow[dashrightarrow]{ur}[swap]{\semanticsalt{t}_{\sigma\to\tau}} \\
  \end{tikzcd}
  \]
The rules of the logic translate to properties which we need $p$ to have in order to apply  Lemma~\ref{lemma-extend-gpa}:
(\pointrule) implies that $p$ preserves the order, while (\monorule) and (\antirule) ensure that $p$ preserves polarities.
The induction hypothesis on $\sigma$ includes the statement that $j$ is an order embedding.

We also must check the weak-completeness-like property (\ref{special}) which is a hypothesis of Lemma~\ref{lemma-extend-gpa}.
Suppose that we have 
$f$ and $g$ in
$\PPsyn_{\sigma\to\tau}$
  with $f^\uar\leq  g^\dar$.  The only tagged symbols in that preorder are those in $\GG_{\sigma\to\tau}$, so $f$ and $g$
  are symbols in $\GG_{\sigma\to\tau}$; indeed they come from the original signature.
Let $t,u: \sigma$.
Using the rule (\wcone), $\GammaBox\proveswc f^\uar t:  \tau \leq  g^\dar u:  \tau$.
That is,  $f^\uar t\leq  g^\dar u$ in $\SS = \PPsyn_{\tau}$.
Since $\semanticsalt{\ }_{\tau}$ preserves the order,
$\semanticsalt{ft}_{\tau} \leq \semanticsalt{gu}_{\tau}$ in $\PPsem_{\tau}$.
This means that $p_f(t) \leq p_g(u)$, as required.
 
 We also verify  (\ref{special})  when $f^\dar\leq  g^\uar$.  The work is the same, using
(\wctwo)
 instead of (\wcone).

 Lemma~\ref{lemma-extend-gpa} tells us that $p$ extends to $q:\FF \to \MM^{\LL}$.
 We define $\semanticsalt{\ }_{\sigma\to\tau}:\PPsyn_{\sigma\to\tau}\to \PPsem_{\sigma\to\tau}$ by
\[ \semanticsalt{t}_{\sigma\to\tau} = q_t.\]
For each term $t:\sigma\to\tau$, $q_t$ 
is an element of $\MM^{\LL}$ and hence a function of the right type.
The fact that $q$ preserves polarities and the order implies the same properties of $\semanticsalt{\ }_{\sigma\to\tau}$. 
We have several further verifications.
 
 \paragraph{The applicative family property (\ref{eq-applicative-family}) }
 Let $t \in \PPsyn_{\sigma\to\tau}$ and  let $u\in \PPsyn_\sigma$.
Using the fact from  Lemma~\ref{lemma-extend-gpa}  that $p_t = q_t\o j$, 
 \[
 \semanticsalt{t}_{\sigma\to\tau}(\semanticsalt{u}_{\sigma}) =  q_{t}(\semanticsalt{u}_{\sigma}) =
 q_{t}(j(u)) =
  p_{t} (u)  = \semanticsalt{tu}_{\tau} .
 \]


\paragraph{$\semanticsalt{\ }_{\sigma\to\tau}$  reflects the order}
Suppose that in $\PPsem_{\sigma\to\tau}$, $\semanticsalt{t}_{\sigma\to\tau} \leq \semanticsalt{u}_{\sigma\to\tau}$.
Let $x = \semanticsalt{\Box_{\sigma}}_{\sigma}$. 
Then using the applicative family property which we just showed,
\begin{equation}\label{wejustshowed}
\semanticsalt{t\Box_{\sigma}}_{\tau} = 
\semanticsalt{t}_{\sigma\to\tau}(x) \leq \semanticsalt{u}_{\sigma\to\tau}(x) = \semanticsalt{u\Box_{\sigma}}_{\tau}.
\end{equation}
Since $\semanticsalt{\ }_{\tau}$ reflects order,
in $\PPsyn_{\tau}$, $t\Box_{\sigma} \leq u\Box_{\sigma}$.
Thus, $\GammaBox \proves t\Box_{\sigma} \leq u\Box_{\sigma}$.
By Lemma~\ref{lemma-onebox-wc}, $\GammaBox\proves t\leq u$.
This tells us that $t\leq u$ in $\PPsyn_{\sigma\rightarrow\tau}$.

\paragraph{$\semanticsalt{\ }_{\sigma\to\tau}$  is one-to-one}
Suppose that in $\PPsem_{\sigma\to\tau}$, $\semanticsalt{t}_{\sigma\to\tau} = \semanticsalt{u}_{\sigma\to\tau}$.
As in (\ref{wejustshowed}) above, we have $\semanticsalt{t\Box_{\sigma}}_{\tau} = \semanticsalt{u\Box_{\sigma}}_{\tau}$.
Since  $\semanticsalt{\ }_{\tau}$  is one-to-one, $t\Box_{\sigma} =u\Box_{\sigma}$.  Thus $t=u$.

This concludes the proof of Lemma~\ref{lemma-semanticsalt-wc}.
\end{proof}

Let us complete the proof of 
Theorem~\ref{theorem-main-wc}.
Suppose that $\Gamma\modelswc t^*\leq u^*$.   
By our remarks at the beginning of Section~\ref{section-lemmas-new-constants-weakly-complete},
this assertion 
 holds when we add new symbols to the underlying signature.
Let  $\Nodel = (\semanticsalt{\ }_{\sigma})_{\sigma}$ be the applicative family 
provided by Lemma~\ref{lemma-semanticsalt-wc}.  
Let $\Model$ be the full structure associated to $\Nodel$ using Lemma~\ref{lemma-model-construction}.
Each $\Model_{\sigma}$ is (weakly) complete, since $\Model_{\sigma}$ is the same preorder as $\Nodel_{\sigma}$.
Thus, $\Model\models \Gamma$.
Since the maps  $\semanticsalt{\ }_{\sigma}$ are monotone, $\Model\models \GammaBox$.  
By the assumption in our theorem, $\Model\models t^*\leq u^*$. 
Since all of the maps $\semanticsalt{\ }_{\sigma}$ reflect the order, 
Lemma~\ref{lemma-model-construction} tells us 
that
$\GammaBox\proveswc t^*\leq u^*$.
By 
Lemma~\ref{lemma-no-new-symbols-wc}, $\Gamma\proveswc  t^*\leq u^*$.
\end{proof}

\section{Variations and Extensions}

Our  next section contains results that build on what we saw in the previous section.

\subsection{The logic of full poset structures}
A structure is a \emph{poset structure} if each preorder $\PP_{\sigma}$ is a partially ordered set:
if $p\leq q$ and $q\leq p$, then $p = q$.  For such structures, the following rule is sound:
\[
\infer[\posrule]{fs \leq ft}{s \leq t & t  \leq s }
\]
In this rule, $f\in \FF_{\sigma\to\tau}$ is arbitrary; it need not be tagged $\uar$ or $\dar$.
(When $f$ is tagged either way, (\posrule) is obviously derivable.)
In fact, we have a complete logic of weakly complete poset structures: take the rules in Figures~\ref{fig-rules1} and~\ref{fig-rules2}
and add the (\posrule) rule.   Here are the reasons: 
Every preorder $\QQ$ has an associated poset $\QQ^*$ obtained by taking the quotient $\PP/\!\!\equiv$, where $p\equiv q$ iff $p\leq q\leq p$.
The syntactic preorders $\PPsyn_{\sigma}$ determined by a set $\Gamma$ in the logic with (\posrule) may be taken to be a poset; we take the associated 
poset $(\PPsyn_{\sigma})^*$.
 To interpret function symbols on $\PPsyn_{\sigma}$, we need a short well-definedness argument using (\posrule).
We also tag an equivalence class $[f]$ with $\uar$ if some $g\equiv f$ is tagged $\uar$.

Continuing, the constructions of weakly complete preorders which we saw
in Propositions~\ref{prop-complete-functions} and~\ref{prop-CPL} go through when we replace ``preorder'' by ``poset'' in the 
hypothesis and the conclusion.  (In fact, Proposition~\ref{prop-CPL} is a little easier in the poset setting, and it is rather well-known.)

In the proof of Theorem~\ref{theorem-main-wc}, we need to check that some functions are well-defined
on $(\PPsyn_{\sigma})^*$.     Each $p_f$ is well-defined in Lemma~\ref{lemma-semanticsalt-wc}; this comes from (\posrule).  And
in the Extension Lemma~\ref{lemma-extend-gpa}, we observe that if $f\equiv g$, then $q_f \equiv q_g$; this implies that
each $q_f$ is well-defined as a function on $(\PPsyn_{\sigma\to\tau})^*$.

\paragraph{Identities}
Another way to deal with poset structures would be to expand the basic assertions in the language to include
identity statements $t = u$ with the obvious semantics.  (This is also possible even with preordered structures, so we could 
have made this move early on.)  Doing this, we would have the evident rules
\[\infer[\symmrule]{q=p}{p = q}
\qquad\qquad
\infer[\weakrule]{p \leq q}{p = q}
\qquad\qquad
\infer[\posruleprime]{p = q}{p \leq q  & q \leq p }
\]
Here is how the first two rules above are used. 
We need these rules in order to build the syntactic preorders in the first place.  Their elements are equivalence classes
$[t]$ of terms $t$ under the $=$ 
equivalence relation.
Using ($\weakrule$) and ($\posruleprime$), we can derive the reflexivity and transitivity rules for $=$.
We also need them 
 to define the order structure on these classes in such a way that $[t] \leq [u]$ iff $t \leq u$.
This is needed at the very end of the proof of  Theorem~\ref{theorem-main-wc}: 
our previous proof would go from $\semanticsalt{t} \leq \semanticsalt{u}$
to $[t]\leq [u]$.  We need this extra step to know that $\Gamma\proves t \leq u$
(rather than knowing that $\Gamma\proves t' \leq u'$ for some $t'\equiv t$ and $u'\equiv u$.)
The rule ($\posruleprime$) implies (\posrule).   This rule ($\posruleprime$) would
also be used at the very end of the proof of  Theorem~\ref{theorem-main-wc}.
We show that if $\Gamma\modelswc t^* = u^*$, then $\Gamma\proveswc t^* = u^*$.
Our hypothesis easily implies that  $\Gamma\modelswc t^*\leq u^*$ and that  $\Gamma\modelswc u^*\leq t^*$.
By the argument which have seen just above, 
$\Gamma\proveswc t^*\leq u^*$ and   $\Gamma\proveswc u^*\leq t^*$.
Hence using ($\posruleprime$), $\Gamma\proveswc t^* = u^*$.

\subsection{Arrow assertions as conclusions}
\label{section-arrow-assertions}

Up until now, the main assertions in our language have been inequalities between terms of the same type.
The polarity assertions $f^\uar$ and $f^\dar$ were not ``first-class'' (despite what we said at the beginning of Section~\ref{section-proof-system}):
our proof system contained no rules that allowed us to conclude a polarity assertion.
To do this, we 
need to specify the semantics in full structures and to see what must be added to the proof system.
For the semantics, 
suppose we are given a full structure $\Model$ and a symbol $f:\sigma$ of a function type.   Then we 
say
\[
\Model\models f^\uar \mbox{ iff } \semantics{f} \mbox{ is a monotone function}.
\]
The proof theory adds two rules:
\[
\infer[\polrule^{\uar}]{g^\uar}{f^\uar & f \leq g  & g \leq f}
\qquad
\infer[\polrule^{\dar}]{g^\dar}{f^\dar & f \leq g  & g \leq f}
\]
Here, $f$ and $g$ are symbols from the underlying signature $\FF$, and they should be of function type.
The soundness of this rule appears in Theorem~\ref{theorem-main-arrows} below.
When we write $\proveswc$ in the rest of this section, 
we mean provability with the rules in Figures~\ref{fig-rules1} and~\ref{fig-rules2}, together
with the rules   (\polrule${}^\uar$) and   (\polrule${}^\dar$).
 
The completeness proof adds to what we have seen in several ways.
To begin, we need an analog of the construction where we add new symbols $\Box_{\sigma}$.
This time, we add \emph{two} fresh constants.   To ease our notation, we shall elide the type symbols
and simply write these symbols as $\Box_1$ and $\Box_2$.   Given a set $\Gamma$, we write $\Delta$ for the
set of assertions that adds $\Box_1 \leq \Box_2$ for all types.  

We need results on adding these new constants in this way, building on what we saw in 
 Lemmas~\ref{lemma-onebox}, \ref{lemma-richness}, 
 \ref{lemma-onebox-wc}, and~\ref{lemma-richness-wc}.
 In Lemma~\ref{lemma-new-constants-for-new-Lyndon} below, note that some of the assertions
 appear to be weaker than one would want.  Specifically, point (\ref{partthreeDelta}) 
 implies that ``If $\Delta \proveswc t\Box_1 \leq u\Box_2$, then 
 $\Delta \proveswc t \leq u$.''
 The reason why we prefer the more involved statement is that this is what will be used in 
Theorem~\ref{theorem-main-arrows} below.
(An additional support from our formulation is that the converses of all parts of 
Lemma~\ref{lemma-new-constants-for-new-Lyndon} are true as well.)

\begin{lemma} 
\label{lemma-new-constants-for-new-Lyndon}
Let $\Delta$ be defined from $\Gamma$ as above.
\begin{enumerate}
\item If $\Delta \proveswc t\leq \Box_1 $, then $t = \Box_1$. 
If  $\Delta \proveswc \Box_2 \leq t$, then $t = \Box_2$.   
 If $\Delta \proveswc t\leq \Box_2 $, then either $t = \Box_1$ or $t=\Box_2$.
   If $\Delta \proveswc \Box_1\leq t$, then either $t = \Box_1$ or $t=\Box_2$.
 \label{partoneDelta} 

\item If $\Delta\proveswc t\leq v\Box_1$, then one of the following holds:
\begin{enumerate}
\item
 there is some $u$ such that $t=u\Box_1$
and $\Delta\proveswc u \leq v$, or else 
\item  there are $u$ and $g^\dar$ such that $t = u\Box_2$ and
$\Delta\proveswc u \leq g^\dar \leq v$;
or 
\item 
there is a term  $s:\sigma$, 
   and constants $f,g:\sigma\to\tau$ such that $\Delta\proveswc t\leq f s$, and
   $f\lequardar  g\leq v$. 
\end{enumerate}
There are also similar facts when $\Delta\proveswc v\Box_1\leq t$, 
$\Delta\proveswc t\leq v\Box_2$, and $\Delta\proveswc v\Box_2\leq t$.
 \label{parttwoDelta} 
\item If $\Delta \proveswc t\Box_1 \leq u\Box_1$, then $\Delta \proveswc t \leq u$.    
If $\Delta \proveswc t\Box_2 \leq u\Box_2$, then $\Delta \proveswc t \leq u$.   
 \label{partthreeDelta} 

 If  $\Delta \proveswc t\Box_1 \leq u\Box_2$, then there is a function symbol $f^\uar$ such that
$\Delta \proveswc t \leq f^\uar\leq u$.

 If  $\Delta \proveswc t\Box_2 \leq u\Box_1$, then there is a function symbol $f^\dar$ such that
$\Delta \proveswc t \leq f^\dar\leq u$.

 \item If $\Delta \proveswc f^\uar$, then $\Gamma\proveswc f^\uar$; similarly for $\dar$.
 \label{partsixDelta} 
\end{enumerate}
\end{lemma}

\begin{proof}
Each assertion in part (\ref{partoneDelta}) is a straightforward induction.

Part (\ref{parttwoDelta}) also is proved by four straightforward inductions.
For one step, suppose that $\Delta \proveswc t \leq v \Box_1$ with a proof that ends with (\antirule) using
$\Box_1\leq \Box_2$.  Then there is an antitone function symbol from the signature, say $g^\dar$, such that
$v\Box_1 = g^\dar\Box_1$ and $t = g^\dar \Box_2$.
So in this case,  we have $u = g^\dar = v$.

Part (\ref{partthreeDelta}) is  proved by simultaneous induction
on the number of transitivity steps in derivations.
Here is the transitivity step in the first assertion.
  Suppose that the root uses (\transrule), say
  \[ \infer[\transrule]{t\Box_1 \leq u\Box_1}{t\Box_1 \leq v & v\leq u\Box_1 }\]
The previous lemma applies to both subproofs, and thus we have $3\times 3 = 9$ cases.
Let us suppose first that above the right subproof we have (a).
There is some $w$ such that $v$ is $w\Box_1$, and $\Delta\proveswc w\leq u$.
The left subproof ends $t\Box_1 \leq w\Box_1$, so by induction hypothesis, $\Delta\proveswc t \leq w$.
And thus $\Delta\proveswc t \leq u$ as well.

Suppose next that above the right subproof we have (b).
We thus have $w$ and $h^\dar$ such that $v = w\Box_2$ and $\Delta \proveswc w \leq h^\dar \leq u$.
Thus, the second subproof concludes $t\Box_1 \leq w\Box_2$.  By  our induction hypothesis, 
there is some $g^\uar$ such that $\Delta \proveswc  t \leq g^\uar \leq w$.   
Hence $\Delta \proveswc t \leq u$, as desired.

The other assertions in part  (\ref{partthreeDelta})  are similar to what we have 
seen, either above or in Lemmas~\ref{lemma-onebox} and~\ref{lemma-richness}.

For part (\ref{partsixDelta}).  We first show that $\Delta \not\proveswc \Box^{\uar}_i$ and  $\Delta \not\proveswc \Box^{\dar}_i$.
The proof is an easy induction on derivations, and it also uses part (\ref{partoneDelta}) of this result.
We next show something stronger than the assertion in  part (\ref{partsixDelta}):
  if $\phi$ is any assertion in this language which has no new $\Box$ symbols
and $\Delta\proveswc \phi$, then $\Gamma\proveswc
 \phi$.  The proof is basically the same as that of Lemma~\ref{lemma-no-new-symbols-wc}:
we observe that the rules (\polrule${}^\uar$) and (\polrule${}^\dar$) cannot eliminate the new $\Box$ symbols:
$f$ in these rules cannot be $\Box^{\uar}_i$ 
since $\Delta \not\proveswc \Box^{\uar}_i$; and 
 if $g$ were $\Box_1$ or $\Box_2$, then since one of the premises is $f\leq g$, we
 would have $f  = \Box^{\uar}_j$ for some $j$ by part (\ref{partoneDelta}).
This again contradicts $\Delta \not\proveswc \Box^{\uar}_j$.
\end{proof}

We turn to our main result on the system.   We state Theorem~\ref{theorem-main-arrows} only mentioning assertions of the form $f^\uar$,
but it also holds for inequality assertions $t^*\leq u^*$, with basically the same statement and proof as in Theorem~\ref{theorem-main-wc}.

\begin{theorem}$\Gamma\modelswc f^{*\uar}$ iff $\Gamma\proveswc f^{*\uar}$, and similarly for $\dar$.
\label{theorem-main-arrows}
\end{theorem}

\begin{proof}
Here is the soundness half. 
Let $\Model$ be a full hierarchy, and assume the hypotheses of  (\polrule${}^\uar$).
Let the type involved be the (function) type $\sigma$. Since $f^\uar$, we know that  $\semantics{f}$ is monotone.  
Also,
$\semantics{f}\leq \semantics{g}\leq \semantics{f}$.  
Thus $ \semantics{g}$ is also monotone, as desired.
The argument for  (\polrule${}^\dar$) is similar.

We turn to the completeness of the logic.   Suppose that $\Gamma\modelswc f^\uar$.
Starting from $\Gamma$, we form a theory $\Delta$ as mentioned earlier: 
for each type $\sigma$, we add two fresh constants $\Box_1$ and $\Box_2$ to the
signature, and the assertion $\Box_1\leq \Box_2$ to the theory.
We need to know that $\Delta\modelswc f^\uar$, and this is straightforward by considering reducts:
every full model of $\Delta$ is (after throwing away the interpretations of the new symbols) a model of $\Gamma$,
and so the interpretation of $f$ will be monotone.

At this point we are going to replay the proof of Theorem~\ref{theorem-main-wc} and dwell only on the changes that are to be made.
Form $\PPsyn_{\sigma}$ and $\PPsem_{\sigma}$ as before, except that now we regard them as \emph{polarized} preorders in the evident way:
in  $\PPsyn_{\sigma}$  we use provability from $\Gamma$ to determine the polarities, and in  $\PPsem_{\sigma}$ we use the monotonicity/antitonicity of actual functions.

In  Lemma~\ref{lemma-semanticsalt-wc} we amend the statement to also say that
for a function type $\sigma$,
 $\semanticsalt{\ }_{\sigma}$
reflects polarities.  (This function \emph{preserves} polarities, since this is part of the definition of an applicative family of interpretations.)
 We therefore must check that if $\semanticsalt{g}$ is monotone, then $\Delta\proveswc g^\uar$.
 Since $\semanticsalt{\ }_{\sigma}$ is monotone (by induction hypothesis), $\semanticsalt{\Box_1} \leq \semanticsalt{\Box_2}$.
 By monotonicity, $\semanticsalt{g}(\semanticsalt{\Box_1}) \leq \semanticsalt{g}(\semanticsalt{\Box_2})$ in $\PPsem_{\tau}$.
 Since $\semanticsalt{\ }_{\tau}$ reflects order,
we get that $\Delta \proveswc g \Box_1 \leq g\Box_2$.   By Lemma~\ref{lemma-new-constants-for-new-Lyndon}(\ref{partthreeDelta})
with $t= g=u$,
there is a symbol $h$ in the underlying signature which is tagged $\uar$ such that $\Delta\proveswc g\leq h^\uar \leq g$.
By (\polrule${}^\uar$), $\Delta \proveswc g^\uar$.  This concludes the changes in  Lemma~\ref{lemma-semanticsalt-wc}.

To resume and complete the proof of our theorem, suppose that $f^*$ is a symbol of function type and $\Gamma\modelswc f^{*\uar}$.
Consider the full model $\Model$ whose preorders are $\PPsem_{\sigma}$ with interpretations given by   Lemma~\ref{lemma-semanticsalt-wc}.
Since those interpretations are monotone, $\Model\modelswc \Delta$.   Thus, $\semantics{f^*} = \semanticsalt{f^*}$ is monotone.  
Since $\semanticsalt{\ }$ reflects polarities, $\Delta\proveswc f^{*\uar}$.
In view of Lemma~\ref{lemma-new-constants-for-new-Lyndon}(\ref{partsixDelta}), $\Gamma\proveswc f^{*\uar}$.
\end{proof}

The result in this section may be recast as a ``Lyndon-type'' theorem.
Statements like this may be found in~\cite{IcardMoss21}
and~\cite{Tune2016}.   But in both of these cases, the hypotheses are different, the languages include variables and
abstraction but no polarity assertions, 
and the class of models includes more general models rather than just the full structures.
But all of these are of the form ``semantically monotone implies $\uar$;
semantically antitone implies $\dar$''.  

\begin{corollary}
Fix a set $\Gamma$.
  Let $t$ be a term of function type, and assume that  $\semantics{t}$ is monotone
in all (full) models of $\Gamma$; and also that for some function symbol $f$ from the underlying signature, $\Gamma\proveswc f \leq t \leq f$.
Then there is a symbol $f$ 
with this property such that $\Gamma\proveswc f^\uar$.
\label{corollary-Lyndon}
\end{corollary}

\subsection{The logic of higher-order applicative terms and equality}

For our last variation, we consider  higher-order applicative terms and equality.
In other words,
we abandon the order structure entirely and consider
 the simply typed lambda calculus without variables or abstraction.  The statements of interest are
 identities between terms of the same type, and the semantic notion is given by (\ref{eq-consequence}).
 For the logic, we take the reflexive, symmetric, and transitive laws for $=$, and also the congruence rule
 for application
 \[
 \infer[\congrule]{tu = t'u'}{ t = t' & u = u'}
 \]
This logic is complete, and we sketch the proof.

First, we need lemmas on constants in both the semantics and the proof theory.
Let $\Gamma$ be a set of identity assertions between terms, and let $\GammaBox$ add fresh constants of every type.
In the syntax, the lemma would say that if
$\GammaBox \proves t\Box = u\Box$, then $\Gamma \proves t = u$.
In the semantics, we would want to know that for all assertions $t^* = u^*$ in the language of $\Gamma$, 
if $\Gamma\models t^* = u^*$, then also $\GammaBox\models t^* = u^*$.
 
Suppose that  $\Gamma\models t^* = u^*$.
As we have argued,  we have $\GammaBox\models t^* = u^*$.
For each type $\sigma$, let $\Xsyn_{\sigma}$ be the set of terms of type $\sigma$ in the expanded signature, 
modulo the equivalence relation 
$R(t,u) \leftrightarrow \GammaBox\proves t = u$.  So the elements $\Xsyn_{\sigma}$ are equivalence classes $[t]$ of terms.

We build a full hierarchy of sets $(\Xsem_{\sigma})$ in the evident way, by taking $\Xsem_{\beta} = \Xsyn_{\beta}$ for base types $\beta$,
and for other types, $\Xsem_{\sigma\to\tau} = (\Xsem_{\tau})^{\Xsem_\sigma}$.

We now prove that there is a family of injective maps $\semanticsalt{\ }_{\sigma} : \Xsyn_{\sigma}\to\Xsem_{\sigma}$
with the property that $\semanticsalt{[tu]}_{\tau} = \semanticsalt{[t]}_{\sigma\to\tau}(\semanticsalt{[u]}_{\sigma})$.
When $\sigma$ is a base type, we take $\semanticsalt{\ }_{\sigma} $ to be the identity.
Suppose we are given $\semanticsalt{\ }_{\sigma}$ and $\semanticsalt{\ }_{\tau}$ with the desired properties, 
and we wish to define $\semanticsalt{\ }_{\sigma\to\tau}$.
The definition is 
\[
\semanticsalt{[t]}_{\sigma\to\tau}(x) = 
\biggl\{
\begin{array}{ll}
\semanticsalt{[tu]}_{\tau} & \mbox{if for some (unique) $u:\sigma$, $x = \semanticsalt{[u]}_{\sigma}$} \\
\semanticsalt{[\Box_{\tau}]}_{\tau} & \mbox{if there is no such term $u:\sigma$}\\
\end{array}
\biggr.
\]
where $t:\sigma\to\tau$ is a term and $x\in\PPsem_{\sigma}$.
In the bottom line, $\semanticsalt{[\Box_{\tau}]}_{\tau}$  is the only element of $\Xsem_{\tau}$  that is sure to exist; no features of it are important.
Here is the verification of the uniqueness of $x$ in the top line: if $\semanticsalt{[u]}_{\sigma} = x =\semanticsalt{[u']}_{\sigma}$,
then since $\semanticsalt{\ }_{\sigma}$ is injective (by our inductive assumption), $[u] = [u']$.  
We also check that the top line of this definition is independent of the choice of representatives of the classes $[t]$ and  $[u]$.
For if  $\Gamma\proves t = t'$ and also
$\Gamma\proves u = u'$, then also $\Gamma\proves tu = t'u'$ by (\congrule).
Hence $[tu] = [t'u']$.
It remains to check that $\semanticsalt{\ }_{\sigma\to\tau}$ is injective.   Suppose that $\semanticsalt{[t]}_{\sigma\to\tau} = \semanticsalt{[t']}_{\sigma\to\tau}$.
 Then 
 \[
   \semanticsalt{t\Box_{\sigma}}_{\tau} =
 \semanticsalt{[t]}_{\sigma\to\tau}(\semanticsalt{[\Box_{\sigma}]}_{\sigma})  = 
  \semanticsalt{[t']}_{\sigma\to\tau}(\semanticsalt{[\Box_{\sigma}]}_{\sigma}) = \semanticsalt{t'\Box_{\sigma}}_{\tau}
 \]
 So by injectivity of $\semanticsalt{ \ }_{\tau}$, $\GammaBox \proves t\Box = t'\Box$.
 Thus  $\GammaBox \proves  t = t'$, and in other words $[t] = [t']$.
 
 This completes the inductive step of the lemma.   We conclude with a proof of the overall completeness theorem.
 Suppose that $\Gamma\models t^* = u^*$.
 Then also $\GammaBox\models t^* = u^*$.  Let $\Model$ be the full type hierarchy $(\Xsem_{\sigma})_{\sigma}$.
 We have defined maps $\semanticsalt{\ }_{\sigma}: \Xsyn_{\sigma}\to \Xsem_{\sigma}$.
 From these, we interpret the symbols in the original signature by taking $\semantics{t} = \semanticsalt{[t]}_{\sigma}$ for the unique $\sigma$ such that $t:\sigma$.
 As in Lemma~\ref{lemma-model-construction}, for all terms $t:\sigma$, $\semantics{t} = \semanticsalt{[t]}_{\sigma}$.
 It follows that $\Model\models \Gamma$.   By our assumption that $\Gamma\models t^* = u^*$, we see that 
 $\semantics{t^*} = \semantics{u^*}$.  Let $\sigma$ be the type of $t^*$.
 Then  $\semanticsalt{[t^*]}_{\sigma} = \semanticsalt{[u^*]}_{\sigma}$. 
 Since $\semanticsalt{\ }$ is injective, $\GammaBox\proves t^* = u^*$.   By one our our points above, this tells us that 
  $\Gamma\proves t^* = u^*$, as desired.   

\section{Conclusion}

The main results in this paper were the completeness theorems, 
Theorems~\ref{theorem-main-wc} and~\ref{theorem-main-arrows}, and also Corollary~\ref{corollary-Lyndon}.
The theorems suggest that the logical systems in the paper are the ``right'' ones: they are complete for 
the most natural semantics of higher-order applicative terms using a semantics where
 one can declare symbols to be 
interpreted in a monotone or antitone way, and also assert inequalities between terms.
Corollary~\ref{corollary-Lyndon} does something similar, but not for entailment so much as for the 
expressive features of the system.   

There are two ways in which it would be important to go beyond what we did here.

First, 
we 
return to the very start of this paper, the presentation of tonoids as operations defined by types as in (\ref{req4}).
As the reader may have noticed, the type system in this paper was not sufficient to deal with (\ref{req4}).
All of our types were 	``simpler arrows'' $\to$ rather than $\arrowplus$ or $\arrowminus$.  
So we cannot type a function as in (\ref{req4}).
 It is thus of interest to
extend our results to the system where we incorporate monotonicity/antitonicity information
into the type system in a wholehearted manner, at all higher types.
It is possible to formulate a syntax, semantics, and logical system that can handle this extension.
The details are not so simple, and so we shall not enter in to them.  Those details may be found in our paper~\cite[Section 5]{IcardMoss2014}.
We expect that the methods of this paper show that the logical system there is complete for full models, at least when one works over
weakly complete preorders.    

Second, we have not been able to prove the completeness theorem that we are after in this subject, where one considers full 
preorder hierarchies built over arbitrary preorders, without assuming that the base preorders $\PP_{\beta}$ are weakly complete.
This would mean using the most natural logic for higher-order terms in our setting, the rules in Figure~\ref{fig-rules1}.   
In order to motivate the problem, let us review where in our work the assumption of weak completeness actually was used.
Assuming weak completeness gives the additional (\wcrule) rules stated in Figure~\ref{fig-rules2}.   Those rules are not 
sound for all preorders, as shown in Example~\ref{example-up-down}, part (\ref{notsoundalways}).  Yet, they played a key role in 
Lemma~\ref{lemma-semanticsalt-wc}.     Specifically,
Lemma~\ref{lemma-semanticsalt-wc} called on  Lemma~\ref{lemma-extend-gpa}, and in order to apply 
 Lemma~\ref{lemma-extend-gpa}, the logic needed to have the (\wcrule) rules.

Here is a related point: our overall work made critical use of the passage from a preorder $\PP$ to a completion $\PP^*$,
and it also made critical use of the Extension Lemma~\ref{lemma-semanticsalt-wc}. 
  To follow the general proof strategy
of this paper, we seem to require a weaker type of completeness (one that adds fewer points), and a stronger Extension Lemma
(one that works for the original logic).  Getting all of this to work out is a challenge.

\subsection*{Acknowledgements} We are grateful to an anonymous referee for useful comments
and corrections. 
All remaining errors are our own.
We also thank 
Katalin Bimb\'{o} for all her work on this volume and other projects which keep alive the memory of Mike Dunn.  

\bibliographystyle{plain}

 \end{document}